\documentclass[11pt,twoside]{article}
\usepackage{amssymb,latexsym,graphicx,hyperref,epstopdf,subcaption}
\normalsize
\newtheorem{theorem}{Theorem}
\newtheorem{lemma}{Lemma}
\newtheorem{corollary}{Corollary}

\newtheorem{operation}{Operation}

\newenvironment{proof}{{\sc Proof. }}{\hfill$\Box$\vspace{0.2in}}
\title{Approximation algorithms for the maximum weight internal spanning tree problem%
\footnote{An extended abstract appears in the {\em Proceedings of COCOON 2017}.}}
\author{Zhi-Zhong Chen\thanks{Division of Information System Design, Tokyo Denki University.
	Hatoyama, Saitama 350-0394, Japan.
	Email: {\tt zzchen@mail.dendai.ac.jp}}
\and
Guohui Lin\thanks{Department of Computing Science, University of Alberta.
	Edmonton, Alberta T6G 2E8, Canada.
	Email: {\tt guohui@ualberta.ca}}
	\thanks{Correspondence author.}
\and
Lusheng Wang\thanks{Department of Computer Science, City University of Hong Kong.
	Tat Chee Avenue, Kowloon, Hong Kong, China.
	Email: {\tt cswangl@cityu.edu.hk}}
\and
Yong Chen\thanks{Institute of Operational Research and Cybernetics, Hangzhou Dianzi University.
	Hangzhou, Zhejiang 310018, China.
	Email: {\tt chenyong@hdu.edu.cn}}
	$^{\ddagger}$
\and
Dan Wang$^\P$}%
\date{\today}
\begin{document}
\maketitle
\begin{abstract}
Given a vertex-weighted connected graph $G = (V, E)$, the maximum weight internal spanning tree (MwIST for short) problem asks for
a spanning tree $T$ of $G$ such that the total weight of the internal vertices in $T$ is maximized.
The unweighted variant, denoted as MIST, is NP-hard and APX-hard, and the currently best approximation algorithm has a proven performance ratio $13/17$.
The currently best approximation algorithm for MwIST only has a performance ratio $1/3 - \epsilon$, for any $\epsilon > 0$.
In this paper, we present a simple algorithm based on a novel relationship between MwIST and the maximum weight matching,
and show that it achieves a better approximation ratio of $1/2$.
When restricted to claw-free graphs, a special case been previously studied, we design a $7/12$-approximation algorithm.

\end{abstract}

\newpage
\section{Introduction}
In the {\em maximum weight internal spanning tree} (MwIST for short) problem,
we are given a vertex-weighted connected graph $G = (V, E)$, where each vertex $v$ of $V$ has a nonnegative weight $w(v)$,
with the objective to compute a spanning tree $T$ of $G$ such that the total weight of the internal vertices in $T$, denoted as $w(T)$, is maximized.
MwIST has applications in the network design for cost-efficient communication~\cite{SW08} and water supply~\cite{BFG13}.

When the vertex weights are uniform, or simply vertex-unweighted,
the problem is referred to as the {\em maximum internal spanning tree} (MIST for short) problem.
MIST is clearly NP-hard because it includes the NP-hard Hamiltonian-path~\cite{GJ79} problem as a special case.
Furthermore, MIST has been proven APX-hard~\cite{LZ14}, suggesting that it does not admit a polynomial-time approximation scheme (PTAS).
In the literature,
much research is done on designing (polynomial-time, if not specified) approximation algorithms for MIST
to achieve the worst-case performance ratio as close to $1$ as possible.

The probably first approximation for MIST is a local search algorithm, which achieves a ratio of $1/2$ and is due to Prieto and Sliper~\cite{PS03}.
Salamon and Wiener~\cite{SW08} later modified slightly Prieto and Sliper's algorithm to make it run faster (in linear-time) while achieving the same ratio of $1/2$.
Besides, two special cases of MIST were considered by Salamon and Wiener~\cite{SW08}:
when restricted to claw-free graphs, they designed a $2/3$-approximation algorithm;
when restricted to cubic graphs, they designed a $5/6$-approximation algorithm.
Later, Salamon~\cite{Sal09} proved that the $1/2$-approximation algorithm in \cite{SW08} actually achieves a performance ratio of
$3/(r+1)$ for the MIST problem on $r$-regular graphs ($r \ge 3$).
Based on local optimization, Salamon~\cite{Sal09b} presented an $O(n^4)$-time $4/7$-approximation algorithm for MIST restricted to graphs without leaves.
The algorithm was subsequently simplified and re-analyzed by Knauer and Spoerhase~\cite{KS09} to run faster (in cubic time),
and it becomes the first improved $3/5$-approximation for the general MIST.
Via a deeper local search strategy than those in \cite{KS09} and \cite{Sal09b},
Li {\it et al.}~\cite{LCW14} presented a further improved approximation algorithm for MIST with ratio $2/3$.
At the same time, Li and Zhu~\cite{LZ14} presented another $2/3$-approximation algorithm for MIST.

Unlike the other previously known approximation algorithms for MIST,
the $2/3$-approximation by Li and Zhu~\cite{LZ14} is based on a simple but crucial observation that
the maximum number of internal vertices in a spanning tree of a graph $G$ can be upper bounded by
the maximum number of edges in a triangle-free $2$-matching (a.k.a. path-cycle cover) of $G$.
The time complexity of this approximation algorithm is dominated by computing the maximum triangle-free $2$-matching,
$O(n m^{1.5} \log n)$, where $n$ and $m$ are the numbers of vertices and edges in $G$, respectively.
Li and Zhu~\cite{LZ14b} claimed that they are able to further improve their design to achieve a $3/4$-approximation algorithm for MIST, of the same time complexity.
Recently, Chen {\it et al.}~\cite{CHW16} gave another $3/4$-approximation algorithm for MIST,
which is simpler than the one in \cite{LZ14b};
and they showed that by applying three more new ideas,
the algorithm can be refined into a $13/17$-approximation algorithm for MIST of the same time complexity.
This is currently the best approximation algorithm for MIST.

The parameterized MIST by the number of internal vertices $k$, and its special cases and variants,
have also been extensively studied in the literature~\cite{PS03, Pri05, PS05, CFG10, FLG12, BFG13, FGS13, LWC15, LJF16}.
The best known kernel for the general problem has a size $2k$, which leads to the fastest known algorithm with running time $O(4^k n^{O(1)})$~\cite{LWC15}.

For the vertex-weighted version, MwIST, Salamon~\cite{Sal09b} designed the first $O(n^4)$-time $1/(2 \Delta - 3)$-approximation algorithm,
based on local search, where $\Delta$ is the maximum degree of a vertex in the input graph.
For MwIST on claw-free graphs without leaves, Salamon~\cite{Sal09b} also designed an $O(n^4)$-time $1/2$-approximation algorithm.
Subsequently, Knauer and Spoerhase~\cite{KS09} proposed the first constant-ratio $1/(3 + \epsilon)$-approximation algorithm for the general MwIST,
for any constant $\epsilon > 0$.
The algorithm is based on a new pseudo-polynomial time local search algorithm,
that starts with a depth-first-search tree and applies six rules to reach a local optimum.
It yields a $1/3$-approximation for MwIST and then is extended to a polynomial time $1/(3 + \epsilon)$-approximation scheme.
The authors also showed that the ratio of $1/3$ is asymptotically tight.

In this paper, we deal with the MwIST problem.
We first prove a novel relationship between the total weight of the internal vertices in a spanning tree of the given vertex-weighted graph
and the maximum weight matching of an edge-weighted graph, that is constructed out of the given vertex-weighted graph.
Based on this relationship, we present a simple $1/2$-approximation algorithm for MwIST;
this ratio $1/2$ significantly improves upon the previous known ratio of $1/3$.
When restricted to claw-free graphs, a special case previously studied in \cite{SW08, Sal09b}, we design a $7/12$-approximation algorithm,
improving the previous best ratio of $1/2$.

\section{The $1/2$-approximation algorithm}
Recall that in the MwIST problem, we are given a connected graph $G = (V, E)$, where each vertex $v$ of $V$ has a nonnegative weight $w(v)$,
with the objective to compute a spanning tree $T$ of $G$ such that the total weight of the internal vertices in $T$, denoted as $w(T)$, is maximized.
We note that for such an objective function, we may assume without loss of generality that every leaf in the given graph $G$ has weight $0$.

We construct an edge-weighted graph based on $G = (V, E)$.
In fact, the structure of the new graph is identical to that of $G$:
the vertex set is still $V$, but instead the vertices have no weights;
the edge set is still $E$, where the weight of each edge $e = \{u, v\}$ is $w(e) = w(u) + w(v)$,
i.e., the weight of an edge is the total weight of its two ending vertices in the original graph.
Since there is no ambiguity when we discuss the edge weights or the vertex weights, the new edge-weighted graph is still referred to as $G$.
The weight of an edge subset refers to the total weight of the edges therein;
while the weight of an acyclic subgraph refers to the total weight of the internal (and those surely will become internal) vertices therein.

Let $M^*$ denote the maximum weight matching of (the edge-weighted graph) $G$,
which can be computed in $O(n \min\{m \log n, n^2\})$-time, where $n = |V|$ and $m = |E|$.

\begin{lemma}
\label{lemma01}
Given a spanning tree $T$ of $G$, we can construct a matching $M$ of $G$ such that $w(T) \le w(M)$.
\end{lemma}
\begin{proof}
We construct $M$ iteratively.
Firstly, we root the tree $T$ at an internal vertex $r$, and all the edges of $T$ are {\em unmarked};
then in every iteration we include into $M$ an unmarked edge $e = \{u, v\}$ of $T$ such that 1) both $u$ and $v$ are internal and
2) $e$ is the closest to the root $r$ measured by the number of edges on the path from $r$ to $e$,
followed by marking all the edges incident at $u$ or $v$.
This way, the total weight of the two internal vertices $u$ and $v$ in the tree $T$ is transferred to $M$ by adding the edge $e$ to $M$.
At the time this iterative procedure stops,
there is no unmarked edge of $T$ connecting two internal vertices,
and thus every internal vertex whose weight has not been transferred to $M$ must be adjacent to at least a leaf each via an unmarked edge.

Next, we iteratively include into $M$ a remaining unmarked edge $e = \{u, v\}$ of $T$,
followed by marking all the edges incident at $u$ or $v$.
This way, the total weight of $u$ and $v$, which is greater than or equal to the weight of the internal vertex between $u$ and $v$,
is transferred to $M$ by adding the edge $e$ to $M$.
At the end of this procedure, $T$ contains no more unmarked edges.
Since leaves in the tree $T$ count nothing towards $w(T)$, we conclude that $w(T) \le w(M)$.
This proves the lemma.
\end{proof}

The following corollary directly follows from Lemma~\ref{lemma01}, stating an upper bound on the total weight of an optimal solution to the MwIST problem.

\begin{corollary}
\label{coro01}
Let $T^*$ denote an optimal (maximum weight internal) spanning tree of $G$.
Then, $w(T^*) \le w(M^*)$.
\end{corollary}

We next start with $M^*$ to construct a spanning tree $T$.
Let the edges of $M^*$ be $e_1, e_2, \ldots, e_k$;
let $e_j = \{a_j, b_j\}$, such that $w(a_j) \ge w(b_j)$, for all $j = 1, 2, \ldots, k$.
Note that there could be vertices of degree $0$ in the spanning subgraph $G[V, M^*]$ with the edge set $M^*$, and there could be edges of weight $0$ in $M^*$;
let $X$ denote the set of such degree-$0$ vertices and the end-vertices of such weight-$0$ edges.
Essentially we do not worry about the degree of any vertex of $X$ in our final tree $T$, since their weights (if any) are not counted towards $w(M^*)$.
This way, we assume without loss of generality that $w(a_j) > 0$ for each edge $e_j$ of $M^*$, and consequently the degree of $a_j$ is $d_G(a_j) \ge 2$,
that is, $a_j$ is adjacent to at least one other vertex than $b_j$ in the graph $G$.
Let $A = \{a_j \mid j = 1, 2, \ldots, k\}$, and $B = \{b_j \mid j = 1, 2, \ldots, k\}$;
note that $V = A \cup B \cup X$.

Let $E^{aa} = E(A, A \cup X)$, i.e., the set of edges each connecting a vertex of $A$ and a vertex of $A \cup X$,
and $E^{ab} = E(A, B)$, i.e., the set of edges each connecting a vertex of $A$ and a vertex of $B$.
Our construction algorithm first computes a maximal acyclic subgraph of $G$, denoted as $H_0$, by adding a subset of edges of $E^{aa}$ to $M^*$.
This subset of edges is a maximum weight spanning forest on $A \cup X$, and it can be computed in $O(|E^{aa}| \log n)$-time via a linear scan.
In the achieved subgraph $H_0$, if one connected component $C$ contains more than one edge,
then the vertex $a_j$ of each edge $e_j = \{a_j, b_j\}$ in $C$ has degree at least $2$, i.e. is internal.
Therefore, the total weight of the internal vertices in the component $C$ is at least half of $w(C \cap M^*)$,
and $C$ is called {\em settled} and left alone by the algorithm.

Our algorithm next considers an arbitrary edge of $M^*$ that is not yet in any settled component, say $e_j = \{a_j, b_j\}$.
In other words, the edge $e_j$ is an {\em isolated} component in the subgraph $H_0$.
This implies that the vertex $a_j$ is not incident to any edge of $E^{aa}$,
and thus it has to be adjacent to some vertex in $B - \{b_j\}$.
If $a_j$ is adjacent to some vertex $b_i$ in a settled component, then this edge $(a_j, b_i)$ is added to the subgraph $H_0$
(the edge $e_j$ is said {\em merged} into a settled component) and the iteration ends.
The updated component remains settled, as $w(a_j) \ge w(e_j) /2$ is saved towards the weight of the final tree $T$.

In the other case, the vertex $a_j$ is adjacent to a vertex $b_i$, such that the edge $e_i = \{a_i, b_i\}$ is also an isolated component in the current subgraph.
After adding the edge $(a_j, b_i)$ to the subgraph, the algorithm works with the vertex $a_i$ exactly the same as with $a_j$ at the beginning.
That is, if $a_i$ is adjacent to some vertex $b_\ell$ in a settled component, then this edge $(a_i, b_\ell)$ is added to the subgraph
(the component that $a_i$ belongs to is {\em merged} into a settled component) and the iteration ends;
if $a_i$ is adjacent to a vertex $b_\ell$, such that the edge $e_\ell = \{a_\ell, b_\ell\}$ is also an isolated component in the current subgraph,
then the edge $(a_i, b_\ell)$ is added to the subgraph, the algorithm works with the vertex $a_\ell$ exactly the same as with $a_j$ at the beginning;
in the last case, $a_i$ is adjacent to a vertex $b_\ell$ inside the current component that $a_i$ belongs to,
then the edge $(a_i, b_\ell)$ is added to the current component to create a cycle,
subsequently the lightest edge of $M^*$ in the cycle is removed,
the iteration ends,
and the current component becomes settled.
We note that in the above last case, the formed cycle in the current component contains at least $2$ edges of $M^*$;
breaking the cycle by removing the lightest edge ensures that at least half of the total weight of the edges of $M^*$ in this cycle (and thus in this component)
is saved towards the weight of the final tree $T$.
Therefore, when the iteration ends, the resulting component is settled.

When the second step of the algorithm terminates, there is no isolated edge of $M^*$ in the current subgraph, denoted as $H_1$,
and each component is acyclic and settled.
In the last step, the algorithm connects the components of $H_1$ into a spanning tree using any possible edges of $E$.
We denote the entire algorithm as {\sc Approx}.

\begin{lemma}
\label{lemma02}
At the end of the second step of the algorithm {\sc Approx}, every component $C$ of the achieved subgraph $H_1$ is acyclic and settled
(i.e., $w(C) \ge w(C \cap M^*) /2$).
\end{lemma}
\begin{proof}
Let $C$ denote a component;
$C \cap M^*$ is the subset of $M^*$, each edge of which has both end-vertices in $C$.

If $C$ is obtained at the end of the first step, then $C$ is acyclic and for every edge $e_j \in C \cap M^*$,
the vertex $a_j$ has degree at least $2$, and thus $w(C) \ge w(C \cap M^*) /2$.

If a subgraph of $C$ is obtained at the end of the first step but $C$ is finalized in the second step,
then $C$ is also acyclic and for every edge $e_j \in C \cap M^*$,
the vertex $a_j$ has degree at least $2$, and thus $w(C) \ge w(C \cap M^*) /2$.

If $C$ is newly formed and finalized in the second step,
then at the time $C$ was formed, there was a cycle containing at least $2$ edges of $M^*$ of which the lightest one is removed to ensure the acyclicity,
and thus the total weight of the internal vertices on this path is at least half of the total weight of the edges of $M^*$ on this cycle.
Also, the vertex $a_j$ of every edge $e_j$ not on the cycle has degree at least $2$.
Thus, $w(C) \ge w(C \cap M^*) /2$.
\end{proof}

\begin{theorem}
\label{thm01}
The algorithm {\sc Approx} is a $1/2$-approximation for the MwIST problem.
\end{theorem}
\begin{proof}
One clearly sees that {\sc Approx} runs in polynomial time, and in fact the running time is dominated by computing the maximum weight matching $M^*$.

From Lemma~\ref{lemma02}, at the end of the second step of the algorithm {\sc Approx},
every component $C$ of the achieved subgraph $H_1$ is acyclic and satisfies $w(C) \ge w(C \cap M^*) /2$.
Since there is no edge of $M^*$ connecting different components of the subgraph $H_1$,
the total weight of the internal vertices in $H_1$ is already at least $w(M^*) /2$, i.e. $w(H_1) \ge w(M^*) /2$.
The last step of the algorithm may only increase the total weight.
This proves that the total weight of the internal vertices of the tree $T$ produced by {\sc Approx} is
\[
w(T) \ge w(H_1) \ge w(M^*) /2 \ge w(T^*) /2,
\]
where the last inequality is by Corollary~\ref{coro01}, which states that $w(M^*)$ is an upper bound on the optimum.
Thus, {\sc Approx} is a $1/2$-approximation for the MwIST problem.
\end{proof}

\section{A $7/12$-approximation algorithm for claw-free graphs}
We present a better approximation algorithm for the MwIST problem on claw-free graphs.
A graph $G = (V, E)$ is called {\em claw-free} if, for every vertex, at least two of its arbitrary three neighbors are adjacent.
We again assume without loss of generality that every leaf in the graph $G$ has weight $0$.
Besides, we also assume that $|V| \ge 5$.

We first present a reduction rule, which is a subcase of Operation 4 in \cite{CHW16},
that excludes certain induced subgraphs of the given graph $G$ from consideration.

\begin{operation}
\label{op1}
If $G$ has a cut-vertex $v$ such that one connected component $C$ of $G - v$ has two, three or four vertices,
then obtain $G_1$ from $G - V(C)$ by adding a new vertex $u$ of weight $0$ and a new edge $\{v, u\}$.

Let $tw(C)$ denote the maximum total weight of the internal vertices in a spanning tree of the subgraph induced on $V(C) \cup v$,
in which $w(v)$ is revised to $0$.
Then there is an optimal spanning tree $T_1$ of $G_1$ of weight $w(T_1)$ if and only if there is an optimal spanning tree $T$ of $G$ of weight $w(T) = w(T_1) + tw(C)$. 
\end{operation}
\begin{proof}
Let $G_c$ denote the subgraph induced on $V(C) \cup v$, that is, $G_c = G[V(C) \cup v]$;
and let $T_c$ denote the spanning tree of $G_c$ achieving the maximum total weight of the internal vertices, that is, $w(T_c) = tw(C)$
($T_c$ can be computed in $O(1)$-time).

Note that in $T_1$, the leaf $u$ must be adjacent to $v$ and thus $w(v)$ is counted towards $w(T_1)$.
We can remove the edge $\{v, u\}$ and $u$ from $T_1$ while attach the tree $T_c$ to $T_1$ by collapsing the two copies of $v$.
This way, we obtain a spanning tree $T$ of $G$, of weight $w(T) = w(T_1) + w(T_c)$ since $w(v)$ is not counted towards $w(T_c)$.

Conversely, for any spanning tree $T$ of $G$, the vertex $v$ is internal due to the existence of $C$.
We may duplicate $v$ and separate out a subtree $T_c$ on the set of vertices $V(C) \cup v$, in which the weight of $v$ is revised to $0$.
This subtree $T_c$ is thus a spanning tree of $G_c$, and every vertex of $V(C)$ is internal in $T$ if and only if it is internal in $T_c$.
We attach the $0$-weight vertex $u$ to the vertex $v$ in the remainder tree via the edge $\{v, u\}$, which is denoted as $T_1$ and becomes a spanning tree of $G_1$;
note that the vertex $v$ is internal in $T_1$.
It follows that $w(T) = w(T_c) + w(T_1)$.
\end{proof}

See Figure~\ref{fig01} for an illustration of the local configurations specified in Operation~\ref{op1}.
When $|V(C)| = 2$ and $G_c$ is a triangle, we refer the configuration as a {\em hanging triangle};
when $|V(C)| = 3$ and $G_c$ contains a length-$4$ cycle, we refer the configuration as a {\em hanging diamond};
when $|V(C)| = 4$ and $G_c$ contains a length-$5$ cycle, we refer the configuration as a {\em hanging pentagon}.
Applying Operation~\ref{op1}, we assume in the sequel that there is no hanging triangle, or hanging diamond, or hanging pentagon in the given graph $G$.

\begin{figure}[htb]
\begin{center}
\begin{subfigure}[b]{0.25\textwidth}
\includegraphics[width=0.7\textwidth]{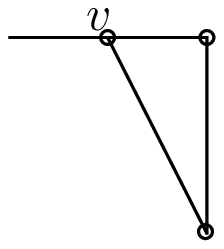}
\caption{A hanging triangle.}
\end{subfigure}%
\hspace{0.05\textwidth}
\begin{subfigure}[b]{0.25\textwidth}
\includegraphics[width=0.65\textwidth]{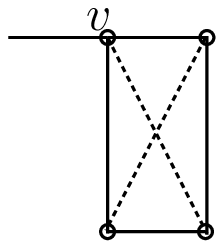}
\caption{A hanging diamond.}
\end{subfigure}%
\hspace{0.05\textwidth}
\begin{subfigure}[b]{0.25\textwidth}
\includegraphics[width=0.8\textwidth]{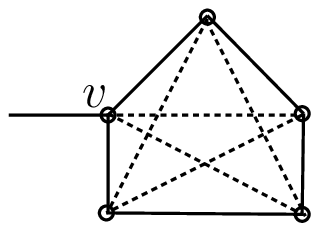}
\caption{A hanging pentagon.}
\end{subfigure}%
\end{center}
\caption{Local configurations of a hanging triangle, a hanging diamond, and a hanging pentagon, specified in Operation~\ref{op1},
	where the dotted edges could be in or not in the graph.\label{fig01}}
\end{figure}

Let $M^*$ denote a maximum weight matching of $G$, which is computed in $O(n \min\{m \log n, n^2\})$-time, where $n = |V|$ and $m = |E|$.
Let the edges of $M^*$ be $e_1, e_2, \ldots, e_k$;
let $e_j = \{a_j, b_j\}$, such that $w(a_j) \ge w(b_j)$, for all $j = 1, 2, \ldots, k$.
For convenience, $a_j$ and $b_j$ are referred to as the {\em head} and the {\em tail} vertices of the edge $e_j$, respectively. 
The same as in the last section, we assume without loss of generality that $w(a_j) > 0$ for each $j$, and consequently the degree of $a_j$ is $d_G(a_j) \ge 2$,
that is, $a_j$ is adjacent to at least one vertex other than $b_j$ in the graph $G$.
Let $A = \{a_j \mid j = 1, 2, \ldots, k\}$, $B = \{b_j \mid j = 1, 2, \ldots, k\}$, and $X = V - (A \cup B)$.

Let 
$E^{aa} = E(A, A)$, i.e., the set of edges each connecting two vertices of $A$,
$E^{ax} = E(A, X)$, i.e., the set of edges each connecting a vertex of $A$ and a vertex of $X$, and
$E^{ab} = E(A, B)$, i.e., the set of edges each connecting a vertex of $A$ and a vertex of $B$, respectively.

Let $M^{aa} \subseteq E^{aa}$ be a maximum cardinality matching within the edge set $E^{aa}$.
We next prove a structure property of the spanning subgraph $G[V, M^* \cup M^{aa}]$, which has the edge set $M^* \cup M^{aa}$.
For an edge $e_j = \{a_j, b_j\}$ of $M^*$, if $a_j$ is not incident to any edge of $M^{aa}$, then $e_j$ is called {\em isolated} in $G[V, M^* \cup M^{aa}]$.

\begin{lemma}
\label{lemma03}
Assume that two edges $e_{j_1} = \{a_{j_1}, b_{j_1}\}$ and $e_{j_2} = \{a_{j_2}, b_{j_2}\}$ of $M^*$ are connected
by the edge $\{a_{j_1}, a_{j_2}\} \in M^{aa}$ in $G[V, M^* \cup M^{aa}]$.
Then there is at most one isolated edge $e_{j_3} = \{a_{j_3}, b_{j_3}\}$ whose head $a_{j_3}$ can be adjacent to $a_{j_1}$ or $a_{j_2}$.
\end{lemma}
\begin{proof}
By contradiction, assume that there are two isolated edges $e_{j_3} = \{a_{j_3}, b_{j_3}\}$ and $e_{j_4} = \{a_{j_4}, b_{j_4}\}$
such that both the vertices $a_{j_3}$ and $a_{j_4}$ are adjacent to $a_{j_1}$ or $a_{j_2}$.
Then from the maximum cardinality of $M^{aa}$, $a_{j_3}$ and $a_{j_4}$ must be both adjacent to $a_{j_1}$ or both adjacent to $a_{j_2}$.
Suppose they are both adjacent to $a_{j_1}$;
from the claw-free property, at least two of $a_{j_2}$, $a_{j_3}$ and $a_{j_4}$ are adjacent, which contradicts the maximum cardinality of $M^{aa}$.
This proves the lemma.
\end{proof}

For an isolated edge $e_{j_3} = \{a_{j_3}, b_{j_3}\}$ whose head is adjacent to an edge $\{a_{j_1}, a_{j_2}\} \in M^{aa}$ (i.e., satisfying Lemma~\ref{lemma03}),
and assuming that $\{a_{j_2}, a_{j_3}\} \in E^{aa}$, we add the edge $\{a_{j_2}, a_{j_3}\}$ to $G[V, M^* \cup M^{aa}]$;
consequently the edge $e_{j_3}$ is no longer isolated.
Let $N^{aa}$ denote the set of such added edges associated with $M^{aa}$.
At the end, the achieved subgraph is denoted as $H_0 = G[V, M^* \cup M^{aa} \cup N^{aa}]$.

\begin{lemma}
\label{lemma04}
In the subgraph $H_0 = G[V, M^* \cup M^{aa} \cup N^{aa}]$,
\begin{itemize}
\parskip=0pt
\item
	every connected component containing more than one edge has either two or three edges from $M^*$,
	with their head vertices connected (by the edges of $M^{aa} \cup N^{aa}$) into a path;
	it is called a {\em type-I} component (see {\em Figure~\ref{fig02a}}) and a {\em type-II} component (see {\em Figure~\ref{fig02b}}), respectively;
\item
	for every isolated edge $e_j = \{a_j, b_j\}$, the head vertex is incident with at least one edge of $E^{ax} \cup E^{ab}$, but with no edge of $E^{aa}$.
\end{itemize}
\end{lemma}
\begin{proof}
The proof directly follows the definition of the subgraph $H_0$ and Lemma~\ref{lemma03}.
\end{proof}

\begin{figure}[htb]
\begin{center}
\begin{subfigure}[b]{0.33\textwidth}
\includegraphics[width=0.55\textwidth]{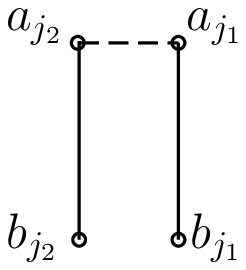}
\caption{A type-I component.\label{fig02a}}
\end{subfigure}%
\hspace{0.05\textwidth}
\begin{subfigure}[b]{0.33\textwidth}
\includegraphics[width=0.7\textwidth]{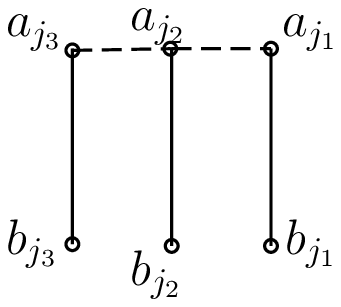}
\caption{A type-II component.\label{fig02b}}
\end{subfigure}%
\end{center}
\caption{The configurations of a type-I component and a type-II component.\label{fig02}}
\end{figure}

The following lemma is analogous to Lemma~\ref{lemma03}.

\begin{lemma}
\label{lemma05}
Any vertex of $X$ can be adjacent to the head vertices of at most two isolated edges in the subgraph $H_0 = G[V, M^* \cup M^{aa} \cup N^{aa}]$.
\end{lemma}
\begin{proof}
By contradiction, assume that $x \in X$ and there are three isolated edges $e_{j_k} = \{a_{j_k}, b_{j_k}\}$, $k = 1, 2, 3$,
in the subgraph $H_0 = G[V, M^* \cup M^{aa} \cup N^{aa}]$, such that the edge $\{a_{j_k}, x\} \in E^{ax}$.
From the claw-free property, at least two of $a_{j_1}$, $a_{j_2}$ and $a_{j_3}$ are adjacent, which contradicts Lemma~\ref{lemma04}.
This proves the lemma.
\end{proof}

For an isolated edge $e_j = \{a_j, b_j\}$ in the subgraph $H_0 = G[V, M^* \cup M^{aa} \cup N^{aa}]$ whose head is adjacent to a vertex $x \in X$
(i.e., satisfying Lemma~\ref{lemma05}), we add the edge $\{a_j, x\}$ to $H_0$;
consequently the edge $e_j$ is no longer isolated.
Let $N^{ax}$ denote the set of such added edges associated with $X$.
At the end, the achieved subgraph is denoted as $H_1 = G[V, M^* \cup M^{aa} \cup N^{aa} \cup N^{ax}]$.

\begin{lemma}
\label{lemma06}
In the subgraph $H_1 = G[V, M^* \cup M^{aa} \cup N^{aa} \cup N^{ax}]$,
\begin{itemize}
\parskip=0pt
\item
	every connected component of $H_0$ containing more than one edge remains unchanged in $H_1$;
\item
	every connected component containing a vertex $x$ of $X$ and some other vertex has either one or two edges from $M^*$,
	with their head vertices connected (by the edges of $N^{ax}$) to the vertex $x$;
	it is called a {\em type-III} component (see {\em Figure~\ref{fig03a}}) and a {\em type-IV} component (see {\em Figure~\ref{fig03b}}), respectively;
\item
	for every isolated edge $e_j = \{a_j, b_j\}$, the head vertex is incident with at least one edge of $E^{ab}$, but with no edge of $E^{aa} \cup E^{ax}$.
\end{itemize}
\end{lemma}
\begin{proof}
The proof directly follows the definition of the subgraph $H_1$ and Lemmas~\ref{lemma04} and \ref{lemma05}.
\end{proof}

\begin{figure}[htb]
\begin{center}
\begin{subfigure}[b]{0.3\textwidth}
\includegraphics[width=0.45\textwidth]{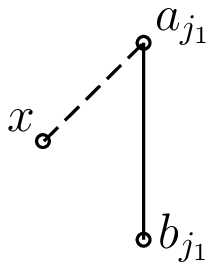}
\caption{A type-III component.\label{fig03a}}
\end{subfigure}%
\hspace{0.05\textwidth}
\begin{subfigure}[b]{0.3\textwidth}
\includegraphics[width=0.75\textwidth]{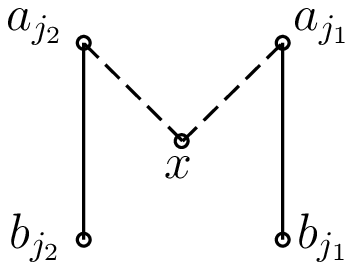}
\caption{A type-IV component.\label{fig03b}}
\end{subfigure}%
\end{center}
\caption{The configurations of a type-III component and a type-IV component.\label{fig03}}
\end{figure}

Let $E^{ab}_0$ denote the subset of $E^{ab}$, to include all the edges $\{a_j, b_\ell\}$ where
both the edges $e_j = \{a_j, b_j\}$ and $e_\ell = \{a_\ell, b_\ell\}$ are isolated in the subgraph $H_1 = G[V, M^* \cup M^{aa} \cup N^{aa} \cup N^{ax}]$.
Let $M^{ab} \subseteq E^{ab}_0$ be a maximum cardinality matching within the edge set $E^{ab}_0$.
Let $H_2 = G[V, M^* \cup M^{aa} \cup N^{aa} \cup N^{ax} \cup M^{ab}]$ be the subgraph obtained from $H_1$ by adding all the edges of $M^{ab}$.
One clearly sees that all the isolated edges in the subgraph $H_1$ are connected by the edges of $M^{ab}$ into disjoint paths and cycles;
while a path may contain any number of isolated edges, a cycle contains at least two isolated edges.
Such a path and a cycle component are called a {\em type-V} component (see Figure~\ref{fig04a}) and a {\em type-VI} component (see Figure~\ref{fig04b}),
respectively.

\begin{figure}[htb]
\begin{center}
\begin{subfigure}[b]{0.35\textwidth}
\includegraphics[width=0.9\textwidth]{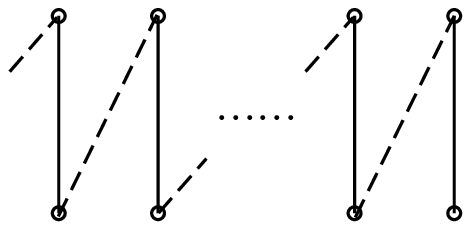}
\caption{A type-V component.\label{fig04a}}
\end{subfigure}%
\hspace{0.05\textwidth}
\begin{subfigure}[b]{0.35\textwidth}
\includegraphics[width=0.8\textwidth]{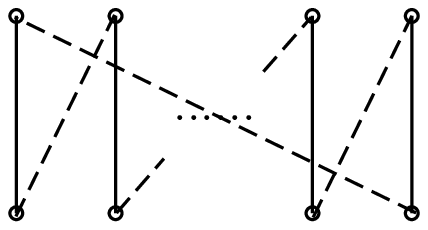}
\caption{A type-VI component.\label{fig04b}}
\end{subfigure}%
\end{center}
\caption{The configurations of a type-V component and a type-VI component.\label{fig04}}
\end{figure}

Note that in a type-V component, there is exactly one head vertex of degree $1$ and there is exactly one tail vertex of degree $1$.
We assume that for the tail vertex in a type-V component, it is not adjacent to the head of any other edge (via an edge of $E^{ab}$) in the same component;
otherwise, through an edge exchange, the component is decomposed into a smaller type-V component and a new type-VI component.

\begin{lemma}
\label{lemma07}
In the subgraph $H_2 = G[V, M^* \cup M^{aa} \cup N^{aa} \cup N^{ax} \cup M^{ab}]$,
for every type-V component, the degree-$1$ head vertex is adjacent (via an edge of $E^{ab}$) to the tail vertex of an edge in a type-I, -II, -III, or -IV component;
on the other hand, the tail vertex of every edge in a type-I, -II, -III, or -IV component is adjacent to at most one such head vertex.
\end{lemma}
\begin{proof}
We first show that the degree-$1$ head vertex in a type-V component $C$, denoted as $a_j$,
cannot be adjacent to the tail of any edge in another type-V or a type-VI component $C'$.
By contradiction, assume $\{a_j, b_\ell\} \in E^{ab}$ and $e_\ell$ is in $C'$.
If the tail $b_\ell$ is already incident to some edge of $M^{ab}$, say $\{a_i, b_\ell\}$,
then by the claw-free property at least two of $a_i, a_j, a_\ell$ must be adjacent, contradicting the fact that they are all isolated in the subgraph $H_1$.
In the other case, the tail $b_\ell$ is the tail vertex of $C'$ (which is a type-V component too),
then it violates the maximum cardinality of $M^{ab}$ since $\{a_j, b_\ell\} \in E^{ab}$ can be added to increase the size of $M^{ab}$.
This proves the first half of the lemma.

The second half can be proven by a simple contradiction using the claw-free property of the graph.
\end{proof}

Subsequently, every type-V component $C$ is connected to a type-I, -II, -III, or -IV component $C'$,
via the edge between the degree-$1$ head vertex of $C$ and the tail vertex of an edge in $C' \cap M^*$.
This way, the degree-$1$ tail vertex of $C$ takes up the role of ``the tail vertex'' of the edge in $C' \cap M^*$,
to become a tail vertex in the newly formed bigger component.
For simplicity, the type of the component $C'$ is passed to the newly formed bigger component.
Denote this set of newly added edges as $N^{ab}$, which is a subset of $E^{ab} - M^{ab}$.
The achieved subgraph is denoted as $H_3 = G[V, M^* \cup M^{aa} \cup N^{aa} \cup N^{ax} \cup M^{ab} \cup N^{ab}]$.

\begin{lemma}
\label{lemma08}
In the subgraph $H_3 = G[V, M^* \cup M^{aa} \cup N^{aa} \cup N^{ax} \cup M^{ab} \cup N^{ab}]$,
\begin{itemize}
\parskip=0pt
\item
	there is no isolated type-V component;
\item
	the head vertex of every edge of $M^*$ has degree at least $2$.
\end{itemize}
\end{lemma}
\begin{proof}
The first half of the lemma follows from Lemma~\ref{lemma07};
the second half holds since there is no more isolated type-V component, which is the only type of component containing a degree-$1$ head vertex.
\end{proof}

We next create a set $F$ of edges that are used to interconnect the components in the subgraph $H_3$.
$F$ is initialized to be empty.
By Lemma~\ref{lemma08}, for every type-I, -II, -III, or -IV component $C$ in the subgraph $H_3$, of weight $w(C \cap M^*)$,
it is a tree and the total weight of the internal vertices therein is at least $\frac 12 w(C \cap M^*)$;
for every type-VI component $C$, which is a cycle, by deleting the lightest edge of $C \cap M^*$ from $C$ we obtain a path and
the total weight of the internal vertices in this path is also at least $\frac 12 w(C \cap M^*)$.
In the next five lemmas, we show that every component $C$ in the subgraph $H_3$ can be converted into a tree on the same set of vertices,
possibly with one edge specified for connecting a leaf of this tree outwards,
such that the total weight of the internal vertices (and the leaf, if specified) in the tree is at least $\frac 23 w(C \cap M^*)$.
The specified edge for the interconnection purpose, is added to $F$.
At the end of the process, the component $C$ is called {\em settled}.
A settled component $C$ can be expressed in multiple equivalent ways, for example,
that the total weight of the internal vertices (and the leaf, if specified) in the resulting tree is at least $\frac 23 w(C \cap M^*)$,
or that the total weight of the internal (and the leaf, if specified) vertices in the resulting tree is
at least twice the total weight of the leaves (excluding the specified leaf, if any).

In the sequel, we abuse the vertex notation to also denote its weight in math formulae;
this simplifies the presentation and the meaning of the notation is easily distinguishable.
For estimating the total weight of the internal vertices in a tree in the sequel, we frequently use the following inequality:
\[
\forall w_1, w_2, w_3 \in {\mathbb R}, \ w_1 + w_2 + w_3 - \min\{w_1, w_2, w_3\} \ge 2 \min\{w_1, w_2, w_3\}.
\]

\begin{lemma}
\label{lemma09}
A type-I component in the subgraph $H_3$ can be settled.
\end{lemma}
\begin{proof}
Consider a type-I component $C$ in the subgraph $H_3$.
Recall from Lemma~\ref{lemma08} and the paragraph right above it,
that a general type-I component is an original type-I component (shown in Figure~\ref{fig02a}) augmented with zero to two type-V components.

Let the two original edges of $M^*$ in $C$ be $e_{j_1}$ and $e_{j_2}$ and the two tail vertices be $b^1$ and $b^2$
(which replace $b_{j_1}$ and $b_{j_2}$ to be the tail vertices, respectively) with $w(b^1) \ge w(b^2)$.
The corresponding two head vertices to the tails $b^1$ and $b^2$ are denoted as $a^1$ and $a^2$, respectively.
See Figure~\ref{fig05} for the general configuration of such a component.
We assume that $w(b^1) > 0$, since otherwise $C$ is settled automatically.

\begin{figure}[htb]
\begin{center}
\includegraphics[angle=0,width=0.5\textwidth]{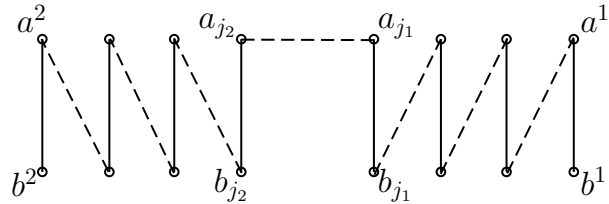}
\end{center}
\caption{The general configuration of a type-I component, with a type-V component adjacent to each of the two tail vertices.\label{fig05}}
\end{figure}

\paragraph{Case 1.}
If $b^1$ is adjacent to a vertex $v$ outside $C$, then we add the edge $\{b^1, v\}$ to $F$;
this settles $C$, since the total weight of the internal vertices in $C$ is at least $a^1 + b^1 + a^2 \ge 3 b^2$
(recall that a vertex notation here represents the weight of the vertex).

We next consider the case where $b^1$ is not adjacent to any vertex outside $C$, and thus it has to be adjacent to some vertex inside $C$.
Note that $C$ is a path with $b^1$ and $b^2$ being its two ending vertices.
Let $v$ denote the vertex adjacent to $b^1$ that is the \underline{\em farthest} to $b^1$ on $C$.
We distinguish this distance $d_C(b^1, v) \ge 2$ and where $v$ locates.

\paragraph{Case 2.}
If $b^1 \ne b_{j_1}$ and $v$ is in the type-V component containing $b^1$,
then by the construction of a type-V component we know that $v$ must be a tail of an edge of $M^*$ (Figures~\ref{fig06a} and \ref{fig06b}),
and thus $d_C(b^1, v)$ is even.
Denote this edge as $e_{j_3} = \{a_{j_3}, b_{j_3}\}$, that is $v = b_{j_3}$.

\begin{figure}[htb]
\begin{center}
\begin{subfigure}[b]{0.45\textwidth}
\includegraphics[width=\textwidth]{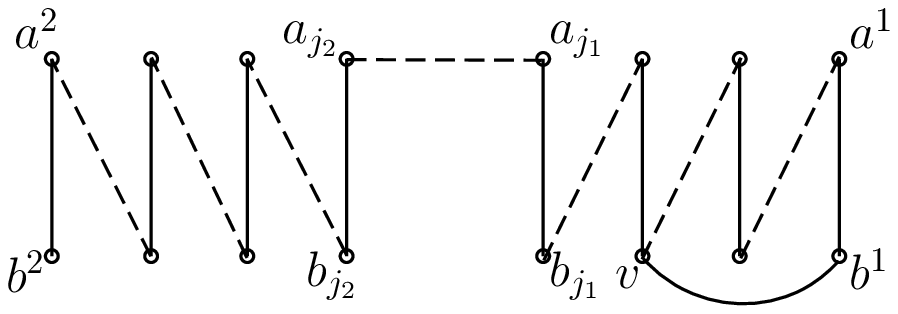}
\caption{$d_C(b^1, v) \ge 4$.\label{fig06a}}
\end{subfigure}%
\hspace{0.04\textwidth}
\begin{subfigure}[b]{0.45\textwidth}
\includegraphics[width=\textwidth]{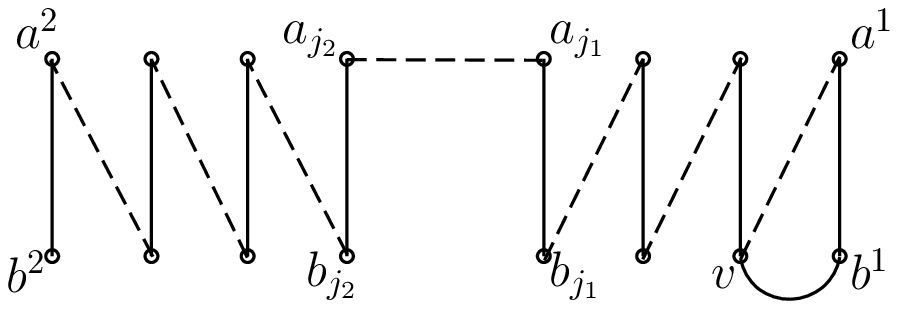}
\caption{$d_C(b^1, v) = 2$.\label{fig06b}}
\end{subfigure}%
\end{center}
\caption{Local configurations corresponding to Case 2, where $v$ is inside the type-V component adjacent to $b_{j_1}$.\label{fig06}}
\end{figure}

{\bf Case 2.1.}
If $d_C(b^1, v) \ge 4$, then denote the head vertex other than $a_{j_3}$ that $b_{j_3}$ is adjacent to as $a_{j_4}$ (see Figure~\ref{fig06a}).
We conclude from the claw-free property that there must be at least an edge among $a_{j_3}, a_{j_4}, b^1$, which contradicts the identity of the type-V component.
Therefore, it is impossible to have $d_C(b^1, v) \ge 4$.

{\bf Case 2.2.}
If $d_C(b^1, v) = 2$, then we conclude that $d_G(b^1) = 2$ and thus $d_G(a^1) \ge 3$ by Operation~\ref{op1} (see Figure~\ref{fig06b}),
i.e. there is at least another edge incident at $a^1$ besides $\{a^1, b_{j_3}\}$ and $\{a^1, b^1\}$.
Denote this neighbor of $a^1$ as $u$.
If $u$ is inside $C$, then $u = b^2$;
in this case, add the edges $\{b^1, b_{j_3}\}$ and $\{a^1, b^2\}$ to $C$
while delete the edge $\{a^1, b_{j_3}\}$ and the lightest among the edges of $C \cap M^*$ from $C$.
This way, the component becomes a tree and thus $C$ is settled.
If $u$ is outside $C$, then we add the edge $\{b^1, b_{j_3}\}$ to $C$ while delete the edge $\{a^1, b_{j_3}\}$ from $C$,
and add the edge $\{a^1, u\}$ to $F$;
this way, the component becomes a tree and thus $C$ is settled.

\paragraph{Case 3.}
If $b^1 \ne b_{j_1}$ and $v = b_{j_1}$, we consider the size of the type-V component containing $b^1$ (see Figure~\ref{fig07}).

\begin{figure}[htb]
\begin{center}
\includegraphics[width=0.45\textwidth]{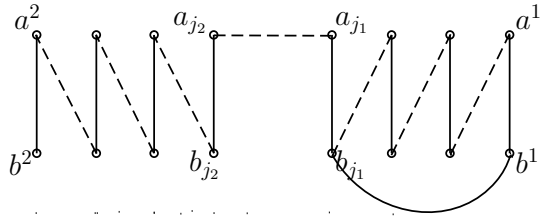}
\end{center}
\caption{Local configurations corresponding to Case 3, where $v = b_{j_1}$.\label{fig07}}
\end{figure}

{\bf Case 3.1.}
If this type-V component contains more than one edge of $M^*$, then by the claw-free property $a_{j_1}$ must be adjacent to $b^1$,
which violates the definition of $v$ being the farthest and thus it is impossible.

{\bf Case 3.2.}
If the type-V component containing $b^1$ has only one edge of $M^*$, which is $\{a^1, b^1\}$,
then we have $d_G(b^1) = 2$, and thus $d_G(a^1) \ge 3$ by Operation~\ref{op1},
i.e. there is at least another edge incident at $a^1$ besides $\{a^1, b_{j_1}\}$ and $\{a^1, b^1\}$.
(The following is the same as in Case 2.2.)
Denote this neighbor of $a^1$ as $u$.
If $u$ is inside $C$, then $u = b^2$;
in this case, add the edges $\{b^1, b_{j_1}\}$ and $\{a^1, b^2\}$ to $C$
while delete the edge $\{a^1, b_{j_1}\}$ and the lightest among the edges of $C \cap M^*$ from $C$.
This way, the component becomes a tree and thus $C$ is settled.
If $u$ is outside $C$, then we add the edge $\{b^1, b_{j_1}\}$ to $C$ while delete the edge $\{a^1, b_{j_1}\}$ from $C$,
and add the edge $\{a^1, u\}$ to $F$;
this way, the component becomes a tree and thus $C$ is settled.

\paragraph{Case 4.}
If $b^1 \ne b_{j_1}$ and $v = a_{j_1}$ (see Figure~\ref{fig08}), then we leave $C$ as it is when $w(b^1) \le w(b_{j_1})$,
or we add the edge $\{a_{j_1}, b^1\}$ to $C$ while delete the edge $\{a_{j_1}, b_{j_1}\}$ from $C$.
The total weight of the internal vertices in the resulting path is at least
$a^2 + a_{j_1} + a^1 + \max\{b^1, b_{j_1}\} \ge 2 (b^2 + \min\{b^1, b_{j_1}\})$,
and thus $C$ is settled.

\begin{figure}[htb]
\begin{center}
\includegraphics[width=0.45\textwidth]{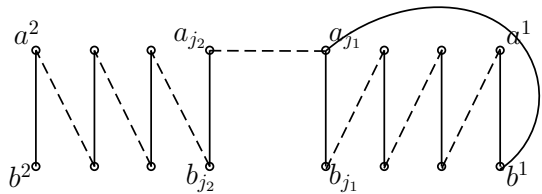}
\end{center}
\caption{Local configurations corresponding to Case 4, where $v = a_{j_1}$.\label{fig08}}
\end{figure}

\paragraph{Case 5.}
If $b^1 \ne b_{j_1}$ and $v = a_{j_2}$ (see Figure~\ref{fig09}), then by the claw-free property there is at least an edge among $b_{j_2}, a_{j_1}, b^1$.
Note that $b_{j_2}$ and $b^1$ cannot be adjacent due to the definition of the vertex $v$ (being the farthest).
If $a_{j_1}$ and $b^1$ are adjacent, then it has been proven in Case 4 that $C$ can be settled.
If $b_{j_2}$ and $a_{j_1}$ are adjacent, then similarly as in Case 4 we either leave $C$ as it is when $w(b^1) \le w(b_{j_1})$,
or add the edges $\{a_{j_2}, b^1\}$ and $\{a_{j_1}, b_{j_2}\}$ to $C$ while delete the edges $\{a_{j_1}, b_{j_1}\}$ and $\{a_{j_2}, b_{j_2}\}$ from $C$;
the total weight of the internal vertices in the resulting path is at least
$a^2 + a_{j_1} + a^1 + \max\{b^1, b_{j_1}\} \ge 2 (b^2 + \min\{b^1, b_{j_1}\})$,
and thus $C$ is settled.

\begin{figure}[htb]
\begin{center}
\includegraphics[width=0.45\textwidth]{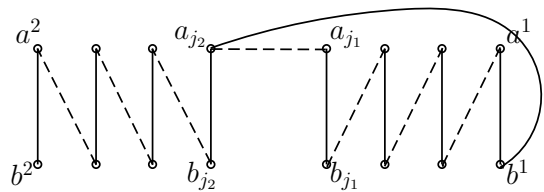}
\end{center}
\caption{Local configurations corresponding to Case 5, where $v = a_{j_2}$.\label{fig09}}
\end{figure}

\paragraph{Case 6.}
If $b^1 \ne b_{j_1}$ and $v = b_{j_2}$ (see Figure~\ref{fig10}), then by the claw-free property $a_{j_2}$ and $b^1$ must be adjacent.

\begin{figure}[htb]
\begin{center}
\includegraphics[width=0.45\textwidth]{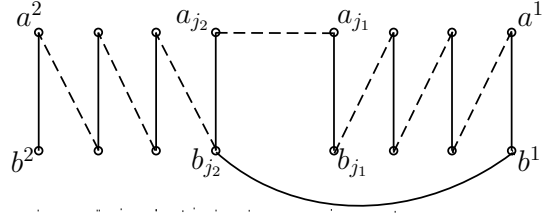}
\end{center}
\caption{Local configurations corresponding to Case 6, where $v = b_{j_2}$.\label{fig10}}
\end{figure}

{\bf Case 6.1.}
If $b^2 \ne b_{j_2}$, then we have three ways to convert $C$ into a path ending at $b^2$:
1) doing nothing to leave $b^1$ as a leaf;
2) adding the edge $\{a_{j_2}, b^1\}$ to $C$ while deleting the edge $\{a_{j_2}, a_{j_1}\}$ to leave $a_{j_1}$ as a leaf;
3) adding the edge $\{b_{j_2}, b^1\}$ to $C$ while deleting the edge $\{a_{j_2}, b_{j_2}\}$ to leave $a_{j_2}$ as a leaf.
Then, the maximum total weight of the internal vertices among these three paths is at least
$a^2 + a^1 + a_{j_1} + a_{j_2} + b^1 - \min\{b^1, a_{j_1}, a_{j_2}\} \ge 2(b^2 + \min\{b^1, a_{j_1}, a_{j_2}\})$.
Thus, $C$ is settled.

{\bf Case 6.2.}
If $b^2 = b_{j_2}$, that is, there is no type-V component adjacent to $b_{j_2}$,
then we add the edge $\{b_{j_2}, b^1\}$ to $C$ while delete the lightest edge of $C \cap M^*$ from $C$ to settle $C$,
because $C \cap M^*$ contains at least three edges.

\paragraph{Case 7.}
If $b^1 \ne b_{j_1}$, $b^2 \ne b_{j_2}$, and $v$ is in the type-V component containing $b^2$,
we distinguish whether $v$ is a head or a tail (see Figure~\ref{fig11}).

\begin{figure}[htb]
\begin{center}
\begin{subfigure}[b]{0.32\textwidth}
\includegraphics[width=\textwidth]{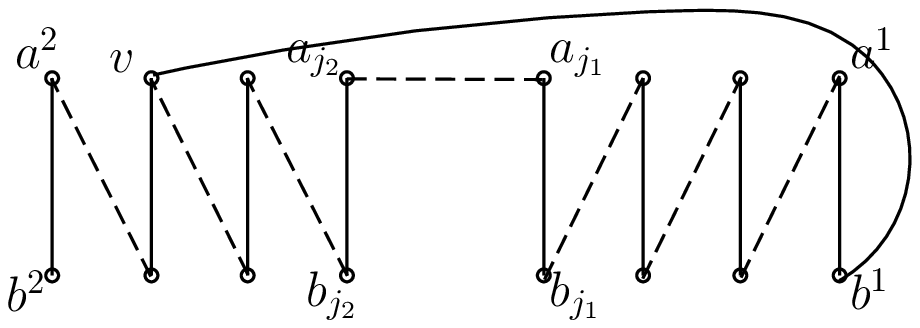}
\caption{$v$ is a head.\label{fig11a}}
\end{subfigure}%
\hspace{0.01\textwidth}
\begin{subfigure}[b]{0.32\textwidth}
\includegraphics[width=\textwidth]{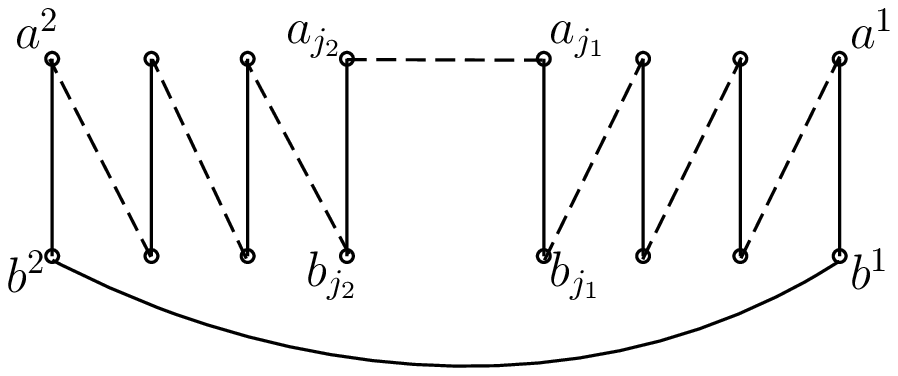}
\caption{$v = b^2$.\label{fig11b}}
\end{subfigure}%
\hspace{0.01\textwidth}
\begin{subfigure}[b]{0.32\textwidth}
\includegraphics[width=0.95\textwidth]{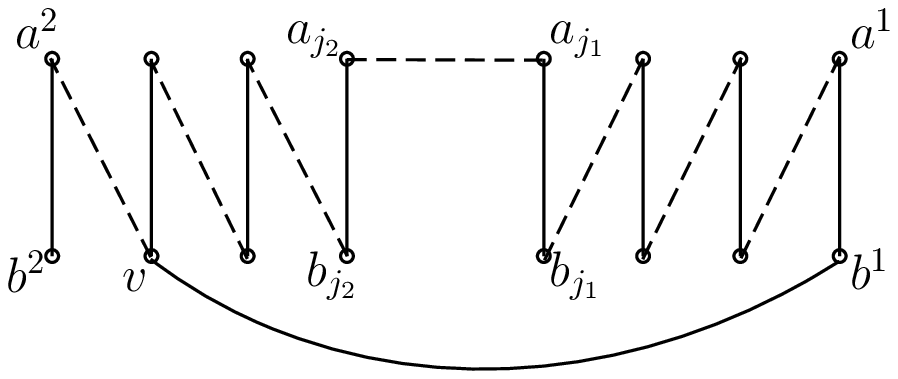}
\caption{$v$ is a tail.\label{fig11c}}
\end{subfigure}%
\end{center}
\caption{Local configurations corresponding to Case 7, where $v$ is inside the type-V component adjacent to $b_{j_2}$.\label{fig11}}
\end{figure}

{\bf Case 7.1.}
If $v$ is a head, say $a_{j_3}$ of the edge $e_{j_3}$ (see Figure~\ref{fig11a}), and assume that $a_{j_3}$ is also adjacent to $b_{j_4}$,
then we do nothing to $C$ to leave $b^1$ as a leaf when $w(b^1) \le w(b_{j_4})$,
or add the edge $\{a_{j_3}, b^1\}$ to $C$ while delete the edge $\{a_{j_3}, b_{j_4}\}$ from $C$ to leave $b_{j_4}$ as a leaf.
Then, the total weight of the internal vertices in the resulting path is at least
$a^2 + a^1 + a_{j_4} + \max\{b^1, b_{j_4}\} \ge 2(b^2 + \min\{b^1, b_{j_4}\})$.
Thus, $C$ is settled.

{\bf Case 7.2.}
If $v$ is a tail, say $b_{j_3}$ of the edge $e_{j_3}$ (see Figures~\ref{fig11b} and \ref{fig11c}).
If $b^2 = b_{j_3}$ (see Figure~\ref{fig11b}),
then we add the edge $\{b^1, b^2\}$ to $C$ makes it a cycle while delete the lightest edge of $C \cap M^*$ to settle $C$
since there are at least four edges in $C \cap M^*$;
if $b^2 \ne b_{j_3}$ (see Figure~\ref{fig11c}), then by the claw-free property and the definition of $v$ we conclude that $b^1$ is adjacent to $a_{j_3}$ too,
and thus the argument in Case 7.1 applies to settle $C$.

In summary, Cases 2-7 together prove that when $b^1 \ne b_{j_1}$, the component $C$ can be settled.
We next consider the situation where $b^1 = b_{j_1}$, that is, there is no type-V component adjacent to $b_{j_1}$ (see Figure~\ref{fig12}).

\begin{figure}[htb]
\begin{center}
\includegraphics[width=0.35\textwidth]{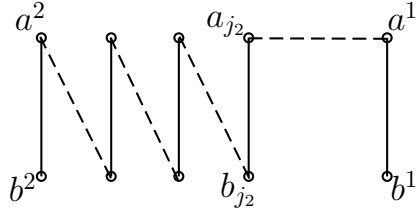}
\end{center}
\caption{A special configuration without a type-V component adjacent to $b_{j_1}$, i.e. $b_{j_1} = b^1$.\label{fig12}}
\end{figure}

\paragraph{Case 8.}
If $b^1 = b_{j_1}$ and $v = a_{j_2}$, then we conclude that $d_G(b^1) = 2$ and thus $d_G(a^1) \ge 3$ by Operation~\ref{op1} (see Figure~\ref{fig13}).

\begin{figure}[htb]
\begin{center}
\includegraphics[width=0.35\textwidth]{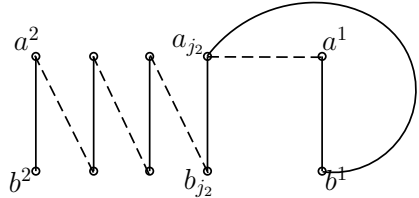}
\end{center}
\caption{Local configurations corresponding to Case 8, where $v = a_{j_2}$.\label{fig13}}
\end{figure}

{\bf Case 8.1.}
If there is a vertex $u$ outside $C$ that is adjacent to $a^1$, then we add the edge $\{a^1, u\}$ to $F$,
add the edge $\{a_{j_2}, b^1\}$ to $C$ while delete the edge $\{a^1, a_{j_2}\}$ from $C$;
this way, $C$ is settled.

{\bf Case 8.2.}
Otherwise by the claw-free property $a^1$ must be adjacent either to $b^2$ or to $b_{j_2}$.

{\bf Case 8.2.1.}
Assuming $b^2 \ne b_{j_2}$, in the former case, we add the edges $\{a_{j_2}, b^1\}$ and $\{a^1, b^2\}$ to $C$
while delete the edge $\{a^1, a_{j_2}\}$ and the lightest edge of $C \cap M^*$ from $C$;
this way, $C$ is settled since there were at least three edges in $C \cap M^*$.
In the latter case, we conclude that $d_G(a_{j_2}) \ge 4$ by Operation~\ref{op1}.
Recursively, if there is a vertex $u$ outside $C$ that is adjacent to $a_{j_2}$, then we add the edge $\{a_{j_2}, u\}$ to $F$,
add the edges $\{a_{j_2}, b^1\}$ and $\{a^1, b_{j_2}\}$ to $C$ while delete the edges $\{a^1, a_{j_2}\}$ and $\{a_{j_2}, b_{j_2}\}$ from $C$;
this way, $C$ is settled.
Otherwise by the claw-free property $a_{j_2}$ must be adjacent to $b^2$;
we add three edges $\{a_{j_2}, b^1\}$, $\{a^1, b_{j_2}\}$ and $\{a_{j_2}, b^2\}$ to $C$
while delete the edges $\{a^1, a_{j_2}\}$ and $\{a_{j_2}, b_{j_2}\}$, and the lightest edge of $C \cap M^*$ from $C$;
this way, $C$ is settled.

{\bf Case 8.2.2.}
Assuming $b^2 = b_{j_2}$, that is, there is no type-V component adjacent to $b_{j_2}$.
We conclude that $d_G(a^2) \ge 4$ and there is a vertex $u$ outside $C$ that is adjacent to $a^2$ by Operation~\ref{op1}.
Thus, we add the edge $\{a^2, u\}$ to $F$,
add the edges $\{a^2, b^1\}$ and $\{a^1, b^2\}$ to $C$ while delete the edges $\{a^1, a^2\}$ and $\{a^2, b^2\}$ from $C$;
this way, $C$ is settled.

\paragraph{Case 9.}
If $b^1 = b_{j_1}$ and $v = b_{j_2}$ (see Figure~\ref{fig14}), we consider two possible scenarios.

\begin{figure}[htb]
\begin{center}
\begin{subfigure}[b]{0.4\textwidth}
\includegraphics[width=0.85\textwidth]{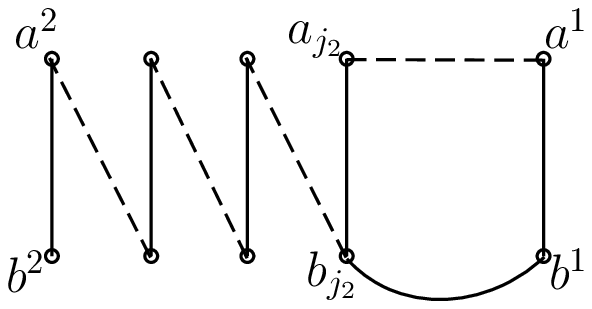}
\caption{$b^2 \ne b_{j_2}$.\label{fig14a}}
\end{subfigure}%
\hspace{0.04\textwidth}
\begin{subfigure}[b]{0.4\textwidth}
\includegraphics[width=0.45\textwidth]{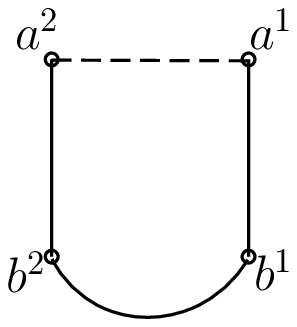}
\caption{$b^2 = b_{j_2}$.\label{fig14b}}
\end{subfigure}%
\end{center}
\caption{Local configurations corresponding to Case 9, where $v = b_{j_2}$.\label{fig14}}
\end{figure}

{\bf Case 9.1.}
If $b^2 \ne b_{j_2}$ (see Figure~\ref{fig14a}), then we conclude from the claw-free property and the definition of $v$ that $b^1$ is also adjacent to $a_{j_2}$.
By Operation~\ref{op1}, at least one of $a^1$ and $a_{j_2}$ must be adjacent to another vertex $u$.

{\bf Case 9.1.1.}
If there is a vertex $u$ outside $C$ that is adjacent to $a^1$ ($a_{j_2}$, respectively)
then we add the edge $\{a^1, u\}$ ($\{a_{j_2}, u\}$, respectively) to $F$,
add the edges $\{b_{j_2}, b^1\}$ and $\{a_{j_2}, b^1\}$ ($\{b_{j_2}, b^1\}$, respectively) to $C$
while delete the edges $\{a^1, b^1\}$ and $\{a_{j_2}, b_{j_2}\}$ ($\{a_{j_2}, b_{j_2}\}$, respectively) from $C$;
this way, $C$ is settled.

{\bf Case 9.1.2.}
Otherwise by the claw-free property $u \in \{b^2, b_{j_2}\}$.
If $u = b^2$, then we can settle $C$ by converting $C$ into a cycle on the same set of vertices, followed by deleting the lightest edge of $C \cap M^*$ from $C$.
In the other case, $u = b_{j_2}$ and thus $a^1$ is adjacent to $b_{j_2}$, which by Operation~\ref{op1} is impossible.

{\bf Case 9.2.}
If $b^2 = b_{j_2}$, that is, there is no type-V component adjacent to $b_{j_2}$ (see Figure~\ref{fig14b}).
We conclude that either there is a vertex $u$ outside $C$ that is adjacent to $a^2$,
or $a^2$ is adjacent to $b^1$ and there is a vertex $u$ outside $C$ that is adjacent to $a^1$.
In either case, we add the edge $\{\cdot, u\}$ to $F$, and convert $C$ into a path with $b^2$ and $\cdot$ as two ending vertices;
this way, $C$ is settled.

\paragraph{Case 10.}
If $b^1 = b_{j_1}$, $b^2 \ne b_{j_2}$, and $v$ is in the type-V component containing $b^2$, we distinguish whether $v$ is a head or a tail (see Figure~\ref{fig15}).

\begin{figure}[htb]
\begin{center}
\begin{subfigure}[b]{0.32\textwidth}
\includegraphics[width=\textwidth]{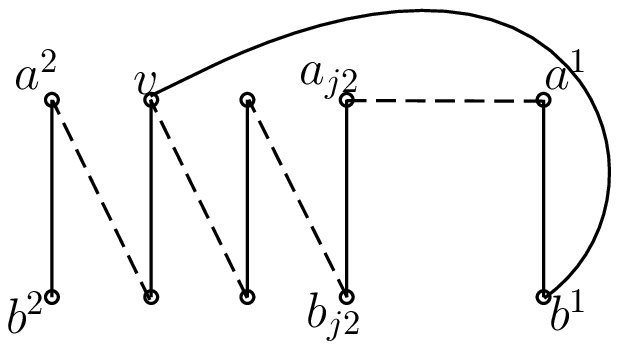}
\caption{$v$ is a head.\label{fig15a}}
\end{subfigure}%
\hspace{0.01\textwidth}
\begin{subfigure}[b]{0.32\textwidth}
\includegraphics[width=\textwidth]{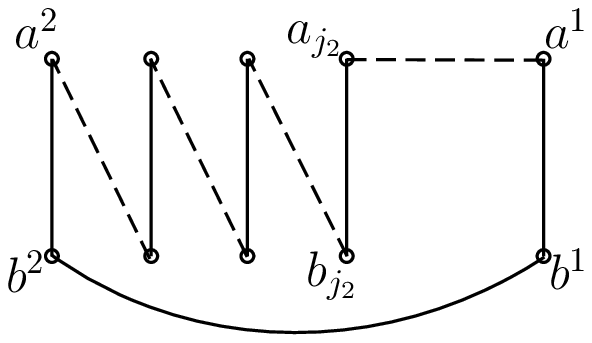}
\caption{$v = b^2$.\label{fig15b}}
\end{subfigure}%
\hspace{0.01\textwidth}
\begin{subfigure}[b]{0.32\textwidth}
\includegraphics[width=0.95\textwidth]{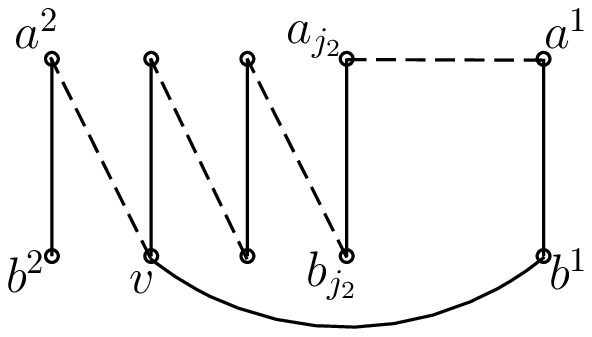}
\caption{$v$ is a tail.\label{fig15c}}
\end{subfigure}%
\end{center}
\caption{Local configurations corresponding to Case 10, where $v$ is inside the type-V component adjacent to $b_{j_2}$.\label{fig15}}
\end{figure}

{\bf Case 10.1.}
If $v$ is a head, say $a_{j_3}$ of the edge $e_{j_3}$ (see Figure~\ref{fig15a}), and assume that $a_{j_3}$ is also adjacent to $b_{j_4}$,
then we do nothing to $C$ to leave $b^1$ as a leaf when $w(b^1) \le w(b_{j_4})$,
or add the edge $\{a_{j_3}, b^1\}$ to $C$ while delete the edge $\{a_{j_3}, b_{j_4}\}$ from $C$ to leave $b_{j_4}$ as a leaf.
Then, the total weight of the internal vertices in the resulting path is at least
$a^2 + a^1 + a_{j_4} + \max\{b^1, b_{j_4}\} \ge 2(b^2 + \min\{b^1, b_{j_4}\})$.
Thus, $C$ is settled.

{\bf Case 10.2.}
If $v$ is a tail, say $b_{j_3}$ of the edge $e_{j_3}$ (see Figures~\ref{fig15b} and \ref{fig15c}).
If $b^2 = b_{j_3}$ (see Figure~\ref{fig15b}),
then we add the edge $\{b^1, b^2\}$ to $C$ makes it a cycle while delete the lightest edge of $C \cap M^*$ to settle $C$
since there are at least three edges in $C \cap M^*$;
if $b^2 \ne b_{j_3}$ (see Figure~\ref{fig15c}), then by the claw-free property and the definition of $v$ we conclude that $b^1$ is adjacent to $a_{j_3}$ too,
and thus the argument in Case 10.1 applies to settle $C$.

All possible cases have be discussed in the above.
The lemma is proven.
\end{proof}

\begin{lemma}
\label{lemma10}
A type-II component in the subgraph $H_3$ can be settled.
\end{lemma}
\begin{proof}
Consider a type-II component $C$ in the subgraph $H_3$.
Recall from Lemma~\ref{lemma08} and the paragraph right above it,
that a general type-II component is an original type-II component (shown in Figure~\ref{fig02b}) augmented with zero to three type-V components.

Let the three original edges of $M^*$ in $C$ be $e_{j_1}, e_{j_2}, e_{j_3}$, with $e_{j_2}$ in the middle (see Figure~\ref{fig02b}),
and the three tail vertices be $b^1, b^2, b^3$
(which replace $b_{j_1}, b_{j_2}, b_{j_3}$, respectively).
We consider the situation where $b^1$ is the heaviest among the three tail vertices
(the other situation is where $b^2$ is the heaviest, and can be similarly discussed).
In the following, discussion for most cases is similar to the cases in the proof of Lemma~\ref{lemma09},
and thus not all details are presented (neither the illustration figures).

\paragraph{Case 1.}
If $b^1$ is adjacent to a vertex $v$ outside $C$, then we add the edge $\{b^1, v\}$ to $F$, which settles $C$,
since the total weight of the internal vertices in $C$ is at least $a^1 + b^1 + a^2 + a^3 \ge 2 (b^2 + b^3)$
(recall that a vertex notation here represents the weight of the vertex).

In the sequel we assume $b^1$ is not adjacent to any vertex outside $C$, and thus it has to be adjacent to some vertex inside $C$.
Let $v$ denote the vertex adjacent to $b^1$ that is the \underline{\em farthest} to $b^1$ on $C$ (tie breaks arbitrarily).
We distinguish this distance $d_C(b^1, v) \ge 2$ and where $v$ locates.

\paragraph{Case 2.}
If $b^1 \ne b_{j_1}$ and $v$ is in the type-V component containing $b^1$, then $v$ must be the tail of an edge of $M^*$ and thus $d_C(b^1, v)$ is even.
Denote this edge as $e_{j_4} = \{a_{j_4}, b_{j_4}\}$, that is $v = b_{j_4}$.

{\bf Case 2.1.}
If $d_C(b^1, v) \ge 4$, then denote the head vertex other than $a_{j_4}$ that $b_{j_4}$ is also adjacent to as $a_{j_5}$.
We conclude from the claw-free property that there must be at least an edge among $a_{j_4}, a_{j_5}, b^1$, which contradicts the identity of the type-V component.
Thus it is impossible to have $d_C(b^1, v) \ge 4$.

{\bf Case 2.2.}
If $d_C(b^1, v) = 2$, then we conclude that $d_G(b^1) = 2$ and thus $d_G(a^1) \ge 3$ by Operation~\ref{op1},
i.e. there is at least another edge incident at $a^1$ besides $\{a^1, b_{j_4}\}$ and $\{a^1, b^1\}$.
Denote this neighbor of $a^1$ as $u$.
If $u$ is inside $C$, then $u \in \{b^2, b^3\}$;
in the case of $u = b^2$ (the argument for $u = b^3$ is identical),
we know that the branch incident at $a_{j_1}$ contains at least three edges of $M^*$ ($e_{j_1}, e_{j_4}, e^1$),
and thus we may add the edges $\{b_{j_4}, b^1\}$ and $\{a^1, b^2\}$ to $C$
while delete the edge $\{a^1, b_{j_4}\}$ and the lightest among the three edges $e_{j_1}, e_{j_4}, e^2$ (which are on the created cycle),
denoted as $e_x = \{a_x, b_x\}$, from $C$.
This way, the component becomes a tree with leaves $b^3, a_x, b_x$.
The total weight of the internal vertices in the tree is at least
$a^3 + (a^1 + b^1) + 2 (a_x + b_x) \ge 2 (b^3 + a_x + b_x)$;
that is, $C$ is settled.

If $u$ is outside $C$, then we add the edge $\{b_{j_4}, b^1\}$ to $C$ while delete the edge $\{a^1, b_{j_4}\}$ from $C$,
and add the edge $\{a^1, u\}$ to $F$;
this way, the component becomes a tree and thus $C$ is settled.

\paragraph{Case 3.}
If $b^1 \ne b_{j_1}$ and $v = b_{j_1}$, we consider the size of the type-V component containing $b^1$.

{\bf Case 3.1.}
If this type-V component contains more than one edge of $M^*$, then by the claw-free property $a_{j_1}$ must be adjacent to $b^1$,
which violates the definition of $v$ being the farthest and thus it is impossible.

{\bf Case 3.2.}
If the type-V component containing $b^1$ has only one edge of $M^*$, which is $\{a^1, b^1\}$;
in this case, we have $d_G(b^1) = 2$, and thus $d_G(a^1) \ge 3$ by Operation~\ref{op1},
i.e. there is at least another edge incident at $a^1$ besides $\{a^1, b_{j_1}\}$ and $\{a^1, b^1\}$.
Denote this neighbor of $a^1$ as $u$.
If $u$ is inside $C$, then $u \in \{b^2, b^3\}$.
In the case of $u = b^2$ (the argument for $u = b^3$ is identical),
we may add the edges $\{b^1, b_{j_1}\}$ and $\{a^1, b^2\}$ to $C$ while delete the edge $\{a^1, b_{j_1}\}$ to form a cycle.
When $b_{j_2} = b^2$, 
we either delete the edge $\{a_{j_2}, b_{j_2}\}$ to leave $b_{j_2}$ as a leaf,
or delete the edge $\{a_{j_2}, a_{j_1}\}$ to leave $a_{j_1}$ as a leaf.
The maximum total weight of the internal vertices between the two trees is at least
$a^3 + a^1 + b^1 + a^2 + \max\{b^2, a_{j_1}\} \ge 2 (b^3 + \min\{b^2, a_{j_1}\})$,
and thus $C$ is settled.
When $b_{j_2} \ne b^2$, we delete the lightest among the three edges $e_{j_1}, e_{j_2}, e^2$ (which are on the created cycle),
denoted as $e_x = \{a_x, b_x\}$, from $C$.
This way, the component becomes a tree with leaves $b^3, a_x, b_x$.
The total weight of the internal vertices in the tree is at least
$a^3 + (a^1 + b^1) + 2 (a_x + b_x) \ge 2 (b^3 + a_x + b_x)$;
that is, $C$ is settled.

If $u$ is outside $C$, then we add the edge $\{b^1, b_{j_1}\}$ to $C$ while delete the edge $\{a^1, b_{j_1}\}$ from $C$,
and add the edge $\{a^1, u\}$ to $F$;
this way, the component becomes a tree and thus $C$ is settled.

\paragraph{Case 4.}
If $b^1 \ne b_{j_1}$ and $v = a_{j_1}$, we consider the size of the type-V component containing $b^1$.

{\bf Case 4.1.} 
If this type-V component contains more than one edge of $M^*$, then denote one edge other than $\{a^1, b^1\}$ as $e_{j_4} = \{a_{j_4}, b_{j_4}\}$.
We add the edge $\{a_{j_1}, b^1\}$ to $C$ while either delete the edge $\{a_{j_1}, b_{j_1}\}$ to have a tree with leaves $b_{j_1}, b^2, b^3$,
or delete the edge $\{a_{j_4}, b_{j_4}\}$ to have a tree with leaves $a_{j_4}, b_{j_4}, b^2, b^3$.
It follows that the maximum total weight of the internal vertices between the two trees is at least
$a^2 + a^3 + a^1 + b^1 + a_{j_1} + \max\{b_{j_1}, a_{j_4} + b_{j_4}\} \ge 2 (b^2 + b^3 + \min\{b_{j_1}, a_{j_4} + b_{j_4}\})$;
therefore, $C$ is settled.

{\bf Case 4.2.} 
If this type-V component contains only one edge of $M^*$, which is $\{a^1, b^1\}$,
then from the claw-free property and the definition of $v$ being the farthest we conclude that $b_{j_1}$ is adjacent to at least one of $a_{j_2}$ and $b^1$.

{\bf Case 4.2.1.} 
Assume $b_{j_1}$ and $a_{j_2}$ are adjacent.
By treating $a_{j_2}$ as the cut-vertex in Operation~\ref{op1},
we conclude that at least one of the three vertices $a_{j_1}, b_{j_1}, a^1$ is adjacent to a vertex $u$ outside of the set $\{a_{j_2}, a_{j_1}, b_{j_1}, a^1, b^1\}$.

{\bf Case 4.2.1.1.} 
If $a_{j_1}$ is adjacent to a vertex $u$ outside the component $C$,
then we add the edges $\{a_{j_2}, b_{j_1}\}$ and $\{a_{j_1}, b^1\}$ to $C$ while delete the edges $\{a_{j_2}, a_{j_1}\}$ and $\{a_{j_1}, b_{j_1}\}$ from $C$
to obtain a tree with leaves $b^3, b^2, a_{j_1}$;
thus, adding the edge $\{a_{j_1}, u\}$ to $F$ settles $C$. 

If $b_{j_1}$ is adjacent to a vertex $u$ outside the component $C$,
then we add the edge $\{a_{j_1}, b^1\}$ to $C$ while delete the edge $\{a_{j_1}, b_{j_1}\}$ from $C$
to obtain a tree with leaves $b^3, b^2, b_{j_1}$;
thus, adding the edge $\{b_{j_1}, u\}$ to $F$ settles $C$. 

If $a^1$ is adjacent to a vertex $u$ outside the component $C$,
then we add the edges $\{a_{j_2}, b_{j_1}\}$ and $\{a_{j_1}, b^1\}$ to $C$ while delete the edges $\{a_{j_2}, a_{j_1}\}$ and $\{b_{j_1}, a^1\}$ from $C$
to obtain a tree with leaves $b^3, b^2, a^1$;
thus, adding the edge $\{a^1, u\}$ to $F$ settles $C$. 

{\bf Case 4.2.1.2.} 
If $a_{j_1}$ is not adjacent to any vertex outside the component $C$, but to some vertices inside $C$, then let $u$ denote the farthest neighbor on $C$
(tie breaks arbitrarily).
There are only four possibilities.

When $u = a_{j_3}$ ($u = b_{j_3}$ and $b_{j_3} \ne b^3$ implying that $a_{j_3}$ and $a_{j_1}$ are adjacent),
we can add the edges $\{a_{j_3}, a_{j_1}\}$, $\{a_{j_2}, b_{j_1}\}$ and $\{a_{j_1}, b^1\}$ to $C$
while delete the edges $\{a_{j_3}, a_{j_2}\}$, $\{a_{j_2}, a_{j_1}\}$ and $\{a_{j_1}, b_{j_1}\}$ from $C$ to obtain a path with leaves $b^3, b^2$;
thus, $C$ is settled. 

When $u = b_{j_2}$, we can add the edges $\{b_{j_2}, a_{j_1}\}$, $\{a_{j_2}, b_{j_1}\}$ and $\{a_{j_1}, b^1\}$ to $C$
while delete the edges $\{a_{j_2}, a_{j_1}\}$, $\{a_{j_2}, b_{j_2}\}$ and $\{a_{j_1}, b_{j_1}\}$ from $C$ to obtain a path with leaves $b^3, b^2$;
thus, $C$ is settled.

When $u = b^3$ (whether $b^3 = b_{j_3}$ or not), we can add the edge $\{b^3, a_{j_1}\}$ to $C$
while delete the edge $\{a_{j_2}, a_{j_1}\}$ from $C$ to obtain a path with leaves $b^2, b^1$;
we also can add the edges $\{b^3, a_{j_1}\}$ and $\{a_{j_1}, b^1\}$ to $C$
while delete the edges $\{a_{j_2}, a_{j_1}\}$ and $\{a_{j_1}, b_{j_1}\}$ from $C$ to obtain a path with leaves $b^2, b_{j_1}$.
It follows that the maximum total weight of the internal vertices between the two paths is at least
$a^2 + a_{j_1} + a^1 + \max\{b_{j_1}, b^1\} \ge 2 (b^2 + \min\{b_{j_1}, b^1\})$;
therefore, $C$ is settled.

When $u = b^2$, it can be shown the same as in the last paragraph by replacing $b^3$ with $b^2$, that $C$ can be settled.

{\bf Case 4.2.1.3.} 
If $a^1$ is not adjacent to any vertex outside the component $C$, but to some vertices inside $C$, then let $u$ denote the farthest neighbor on $C$
(tie breaks arbitrarily).
Note that $u$ cannot be any head vertex, as otherwise it violates the algorithm;
$u$ cannot be any tail vertex either, unless it is $b^2$ or $b^3$.
Therefore there are only two possibilities.

When $u = b^3$ (whether $b^3 = b_{j_3}$ or not), we can add the edges $\{b^3, a^1\}$ and $\{a_{j_1}, b^1\}$ to $C$ while delete the edge $\{a_{j_2}, a_{j_1}\}$.
Then, we either delete $\{a^1, b^1\}$ to obtain a path with leaves $b^2, b^1$,
or delete $\{a^1, b_{j_1}\}$ to obtain a path with leaves $b^2, b_{j_1}$.
It follows that the maximum total weight of the internal vertices between the two paths is at least
$a^2 + a_{j_1} + a^1 + \max\{b_{j_1}, b^1\} \ge 2 (b^2 + \min\{b_{j_1}, b^1\})$;
therefore, $C$ is settled.

When $u = b^2$, it can be shown the same as in the last paragraph by replacing $b^3$ with $b^2$, that $C$ can be settled.

{\bf Case 4.2.1.4.} 
If $b_{j_1}$ is not adjacent to any vertex outside the component $C$, but to some vertices inside $C$, then let $u$ denote the farthest neighbor on $C$
(tie breaks arbitrarily).

When $u \in \{a_{j_3}, b_{j_3}, b_{j_2}\}$ ($u = b_{j_3}$ and $b_{j_3} \ne b^3$ implying that $a_{j_3}$ and $b_{j_1}$ are adjacent),
similarly as in Case 4.2.1.2, we can convert $C$ into a path with leaves $b^3, b^2$ to settle $C$.
In the remaining case, $u$ is inside one of the attached type-V components.

When $u$ is a head, say $a_{j_4}$ of the edge $e_{j_4}$ in the type-V component attached to $b_{j_2}$ (or $b_{j_3}$),
and assume that $a_{j_4}$ is adjacent to $b_{j_5}$ besides $b_{j_4}$,
then we can do nothing to have a tree with leaves $b^3, b^2, b^1$;
or we can add the edge $\{a_{j_1}, b^1\}$ to $C$ while delete the edge $\{a_{j_1}, b_{j_1}\}$ to obtain a tree with leaves $b^3, b^2, b_{j_1}$;
or we can add the edges $\{a_{j_4}, b_{j_1}\}$ and $\{a_{j_1}, b^1\}$ to $C$
while delete the edges $\{a_{j_4}, b_{j_5}\}$ and $\{a_{j_1}, b_{j_1}\}$ to obtain a tree with leaves $b^3, b^2, b_{j_5}$.
The maximum total weight of the internal vertices among the three trees is at least
$a^3 + a^2 + a_{j_1} + b_{j_1} + a_{j_5} + b_{j_5} + a^1 + b^1 - \min\{b^1, b_{j_1}, b_{j_5}\} \ge 2 (b^3 + b^2 + \min\{b^1, b_{j_1}, b_{j_5}\})$,
which settles $C$.

When $u$ is a tail, say $b_{j_4}$ of the edge $e_{j_4}$ in the type-V component attached to $b_{j_2}$ (or $b_{j_3}$), but $b_{j_4} \notin \{b^2, b^3\}$,
then from the claw-free property and the definition of $u$ we conclude that $b_{j_1}$ is also adjacent to $a_{j_4}$.
Thus the argument in the last paragraph applies to settle $C$.

Lastly, when $u = b^2$ (or $b^3$, which can be shown in the same way), we conclude that
$b^2$ is adjacent either to $a_{j_1}$, which is settled in Case 4.2.1.2, or to $a^1$, which is settled in Case 4.2.1.3.

{\bf Case 4.2.2.} 
Assume $b_{j_1}$ is not adjacent to $a_{j_2}$ but to $b^1$.
By treating $a_{j_1}$ as the cut-vertex in Operation~\ref{op1},
we conclude that at least one of the two vertices $b_{j_1}, a^1$ is adjacent to a vertex $u$ outside the set $\{a_{j_1}, b_{j_1}, a^1, b^1\}$.

{\bf Case 4.2.2.1.} 
If $b_{j_1}$ is adjacent to a vertex $u$ outside the component $C$,
then we add the edge $\{a_{j_1}, b^1\}$ to $C$ while delete the edge $\{a_{j_1}, b_{j_1}\}$ from $C$
to obtain a tree with leaves $b^3, b^2, b_{j_1}$;
thus, adding the edge $\{b_{j_1}, u\}$ to $F$ settles $C$. 

If $a^1$ is adjacent to a vertex $u$ outside the component $C$,
then we add the edges $\{b_{j_1}, b^1\}$ and $\{a_{j_1}, b^1\}$ to $C$ while delete the edges $\{a_{j_1}, b_{j_1}\}$ and $\{a^1, b^1\}$ from $C$
to obtain a tree with leaves $b^3, b^2, a^1$;
thus, adding the edge $\{a^1, u\}$ to $F$ settles $C$. 

{\bf Case 4.2.2.2.} 
If $a^1$ is not adjacent to any vertex outside the component $C$, but to some vertices inside $C$, then let $u$ denote the farthest neighbor on $C$
(tie breaks arbitrarily).
Note that $u$ cannot be any head vertex, as otherwise it violates the algorithm;
$u$ cannot be any tail vertex either, unless it is $b^2$ or $b^3$.
Therefore there are only two possibilities.
(This is very similar to Case 4.2.1.3.)

When $u = b^3$ (whether $b^3 = b_{j_3}$ or not),
we can add the edges $\{b^3, a^1\}$ and $\{a_{j_1}, b^1\}$ to $C$ while delete the edge $\{a_{j_2}, a_{j_1}\}$.
Then, we can delete $\{a^1, b^1\}$ to obtain a path with leaves $b^2, b^1$,
or we can delete $\{a^1, b_{j_1}\}$ to obtain a path with leaves $b^2, b_{j_1}$.
It follows that the maximum total weight of the internal vertices between the two paths is at least
$a^2 + a_{j_1} + a^1 + \max\{b_{j_1}, b^1\} \ge 2 (b^2 + \min\{b_{j_1}, b^1\})$;
therefore, $C$ is settled.

When $u = b^2$, it can be shown the same as in the last paragraph by replacing $b^3$ with $b^2$, that $C$ can be settled.

{\bf Case 4.2.2.3.} 
If $b_{j_1}$ is not adjacent to any vertex outside the component $C$, but to some vertices inside $C$, then let $u$ denote the farthest neighbor on $C$
(tie breaks arbitrarily).
Note that $u \ne a_{j_2}$, which is Case 4.2.1.
When $u \in \{a_{j_3}, b_{j_3}, b_{j_2}\}$, similarly as in Case 4.2.1.2, we can convert $C$ into a path with leaves $b^3, b^2$ to settle $C$.
In the remaining case, $u$ is inside one of the attached type-V components.
(This is very similar to Case 4.2.1.4.)

When $u$ is a head, say $a_{j_4}$ of the edge $e_{j_4}$ in the type-V component attached to $b_{j_2}$ (or $b_{j_3}$),
and assume that $a_{j_4}$ is adjacent to $b_{j_5}$ besides $b_{j_4}$,
then we can do nothing to have a tree with leaves $b^3, b^2, b^1$;
or we can add the edge $\{a_{j_1}, b^1\}$ to $C$ while delete the edge $\{a_{j_1}, b_{j_1}\}$ to obtain a tree with leaves $b^3, b^2, b_{j_1}$;
or we can add the edges $\{a_{j_4}, b_{j_1}\}$ and $\{a_{j_1}, b^1\}$ to $C$
while delete the edges $\{a_{j_4}, b_{j_5}\}$ and $\{a_{j_1}, b_{j_1}\}$ to obtain a tree with leaves $b^3, b^2, b_{j_5}$.
The maximum total weight of the internal vertices among the three trees is at least
$a^3 + a^2 + a_{j_1} + b_{j_1} + a_{j_5} + b_{j_5} + a^1 + b^1 - \min\{b^1, b_{j_1}, b_{j_5}\} \ge 2 (b^3 + b^2 + \min\{b^1, b_{j_1}, b_{j_5}\})$,
which settles $C$.

When $u$ is a tail, say $b_{j_4}$ of the edge $e_{j_4}$ in the type-V component attached to $b_{j_2}$ (or $b_{j_3}$), but $b_{j_4} \notin \{b^2, b^3\}$,
then from the claw-free property and the definition of $u$ we conclude that $b_{j_1}$ is also adjacent to $a_{j_4}$.
Thus the argument in the last paragraph applies to settle $C$.

Lastly, when $u = b^2$ (or $b^3$, which can be shown in the same way), we conclude that
$b^2$ is adjacent either to $a_{j_1}$ or to $a^1$, the latter of which is settled in Case 4.2.2.3.
In the remaining case where $b^2$ is adjacent to $a_{j_1}$,
we can add the edge $\{b^2, a_{j_1}\}$ to $C$ while delete the edge $\{a_{j_1}, a_{j_2}\}$ to obtain a path with leaves $b^3, b^1$;
or we can add the edges $\{b^2, a_{j_1}\}$ and $\{b^1, a_{j_1}\}$ to $C$ while delete the edges $\{a_{j_1}, a_{j_2}\}$ and $\{a_{j_1}, b_{j_1}\}$
to obtain a path with leaves $b^3, b_{j_1}$.
The maximum total weight of the internal vertices between the two trees is at least
$a^3 + a^1 + a_{j_1} + \max\{b^1, b_{j_1}\} \ge 2 (b^3 + \min\{b^1, b_{j_1}\})$,
which settles $C$.

\paragraph{Case 5.}
If $b^1 \ne b_{j_1}$ and $v = a_{j_2}$, then by the claw-free property there is at least an edge among $b_{j_2}, a_{j_1}, b^1$.
Note that $b_{j_2}$ and $b^1$ cannot be adjacent due to the definition of the vertex $v$ being the farthest.
If $a_{j_1}$ and $b^1$ are adjacent, then it has been proven in the above Case 4 that $C$ can be settled.
If $b_{j_2}$ and $a_{j_1}$ are adjacent, then we add the edges $\{a_{j_2}, b^1\}$ and $\{a_{j_1}, b_{j_2}\}$ to $C$
while delete the edges $\{a_{j_2}, a_{j_1}\}$ and $\{a_{j_2}, b_{j_2}\}$ from $C$;
this way we obtain a path with two leaves $b^2$ and $b^3$, and thus it settles $C$.

\paragraph{Case 6.}
If $b^1 \ne b_{j_1}$ and $v = b_{j_2}$, then we add the edge $\{b_{j_2}, b^1\}$ to $C$ while delete the edge $\{a_{j_2}, b_{j_2}\}$ from $C$;
this way we obtain a path with two leaves $b^2$ and $b^3$, and thus it settles $C$.

If $b^1 \ne b_{j_1}$ and $v = a_{j_3}$, then we add the edge $\{a_{j_3}, b^1\}$ to $C$ while delete the edge $\{a_{j_3}, a_{j_2}\}$ from $C$;
this way we obtain a path with two leaves $b^2$ and $b^3$, and thus it settles $C$.

If $b^1 \ne b_{j_1}$ and $v = b_{j_3}$, then there are two possible scenarios.
When $b^3 \ne b_{j_3}$, by the claw-free property and the definition of $v$ we conclude that $a_{j_3}$ and $b^1$ must be adjacent,
and the last paragraph shows that $C$ is settled.
When $b^3 = b_{j_3}$, we add the edge $\{b^3, b^1\}$ to $C$
while either delete the edge $\{a^3, a_{j_2}\}$ from $C$ to achieve a path with two leaves $a^3, b^2$,
or delete the edge $\{a_{j_2}, a_{j_1}\}$ from $C$ to achieve a path with leaves $a_{j_1}, b^2$;
we may also do nothing to $C$ which is a tree with leaves $b^1, b^2, b^3$.
Among these three trees, the maximum total weight of the internal vertices is at least
$a^3 + b^3 + a^2 + a_{j_1} + a^1 + b^1 - \min\{a^3, a_{j_1}, b^1 + b^3\} \ge 2 (b^2 + \min\{a^3, a_{j_1}, b^1 + b^3\})$;
thus, $C$ is settled.

\paragraph{Case 7.}
If $b^1 \ne b_{j_1}$, $b^2 \ne b_{j_2}$, and $v$ is in the type-V component containing $b^2$, we distinguish whether $v$ is a head or a tail.
(Note that the same argument applies to $b^3 \ne b_{j_3}$ and $v$ is in the type-V component containing $b^3$.)

{\bf Case 7.1.}
If $v$ is a head, say $a_{j_4}$ of the edge $e_{j_4}$, and assume that $a_{j_4}$ is adjacent to $b_{j_5}$ besides $b_{j_4}$,
then we consider two distinct scenarios.

{\bf Case 7.1.1.}
When $b_{j_5} \ne b_{j_2}$, besides two leaves $b^2$ and $b^3$, we either do nothing to $C$ to leave $b^1$ as a leaf,
or add the edge $\{a_{j_4}, b^1\}$ to $C$ while delete the edge $\{a_{j_4}, b_{j_5}\}$ from $C$ to leave $b_{j_5}$ as a leaf,
or add the edge $\{a_{j_4}, b^1\}$ to $C$ while delete the edge $\{a_{j_2}, b_{j_2}\}$ from $C$ to leave $b_{j_2}$ as a leaf.
Among these three trees, the maximum total weight of the internal vertices is at least
$a^3 + a^2 + a^1 + b^1 + a_{j_2} + b_{j_2} + a_{j_5} + b_{j_5} - \min\{b^1, b_{j_2}, b_{j_5}\} \ge 2 (b^3 + b^2 + \min\{b^1, b_{j_2}, b_{j_5}\})$;
thus, $C$ is settled.

{\bf Case 7.1.2.}
When $b_{j_5} = b_{j_2}$, by the claw-free property and the definition of the vertex $v$ we conclude that $b_{j_2}$ is adjacent to at least one of $b_{j_4}$ and $b^1$.
If $b_{j_2}$ and $b_{j_4}$ are adjacent, then we add the edges $\{a_{j_4}, b^1\}$ and $\{b_{j_2}, b_{j_4}\}$ to $C$
while delete the edges $\{a_{j_2}, b_{j_2}\}$ and $\{a_{j_4}, b_{j_4}\}$ from $C$ to obtain a path with leaves $b^2$ and $b^3$;
if $b_{j_2}$ and $b^1$ are adjacent, then we add the edge $\{b_{j_2}, b^1\}$ to $C$
while delete the edge $\{a_{j_2}, b_{j_2}\}$ from $C$ to obtain a path with leaves $b^2$ and $b^3$.
Thus, $C$ is settled.

{\bf Case 7.2.}
In the other case $v$ is a tail, say $b_{j_4}$ of the edge $e_{j_4}$.

{\bf Case 7.2.1.}
If $b^2 = b_{j_4}$, then besides the leaf $b^3$, we either do nothing to $C$ to leave $b^1, b^2$ as leaves,
or add the edge $\{b^1, b^2\}$ to $C$ and delete the edge $\{a_{j_2}, a_{j_1}\}$ to leave $a_{j_1}$ as a leaf,
or add the edge $\{b^1, b^2\}$ to $C$ and delete the edge $\{a_{j_2}, b_{j_2}\}$ to leave $b_{j_2}$ as a leaf.
Among these three trees, the maximum total weight of the internal vertices is at least
$a^3 + a^1 + a_{j_1} + b_{j_2} + b^1 + b^2 - \min\{b^1 + b^2, a_{j_1}, b_{j_2}\} \ge 2 (b^3 + \min\{b^1 + b^2, a_{j_1}, b_{j_2}\})$;
thus, $C$ is settled.

(In this paragraph, we deal with the case where $b^3$ takes the role of $b^2$ and prove our claim at the beginning of Case 7
``that the same argument applies to $b^3 \ne b_{j_3}$ and $v$ is in the type-V component containing $b^3$''.
The complete assumption of Case 7.2.1 is thus $b^3 \ne b_{j_3}$ and $v = b^3$.
Then, besides the leaf $b^2$, we either do nothing to $C$ to leave $b^1, b^3$ as leaves,
or add the edge $\{b^1, b^3\}$ to $C$ and delete the edge $\{a_{j_2}, a_{j_1}\}$ to leave $a_{j_1}$ as a leaf,
or add the edge $\{b^1, b^3\}$ to $C$ and delete the edge $\{a_{j_2}, a_{j_3}\}$ to leave $a_{j_3}$ as a leaf.
Among these three trees, the maximum total weight of the internal vertices is at least
$a^2 + a^1 + a_{j_1} + a_{j_3} + b^1 + b^3 - \min\{b^1 + b^3, a_{j_1}, a_{j_3}\} \ge 2 (b^2 + \min\{b^1 + b^3, a_{j_1}, a_{j_3}\})$;
thus, $C$ is settled.
In summary, here the vertex $a_{j_3}$ takes up the role of $b_{j_2}$ correspondingly.)

{\bf Case 7.2.2.}
If $b^2 \ne b_{j_4}$, then by the claw-free property and the definition of $v$ we conclude that $b^1$ is adjacent to $a_{j_4}$ too.
Assume that $a_{j_4}$ is adjacent to $b_{j_5}$ besides $b_{j_4}$.
When $b_{j_5} \ne b_{j_2}$, the argument in Case 7.1.1 can be applied to settle $C$;
when $b_{j_5} = b_{j_2}$ (the argument in Case 7.1.2 does not applied to settle $C$ due to the deferent $v$),
besides the leaves $b^2$ and $b^3$, we either do nothing to $C$ to leave $b^1$ as a leaf,
or add the edge $\{b^1, a_{j_4}\}$ to $C$ and delete the edge $\{a_{j_2}, b_{j_2}\}$ to leave $b_{j_2}$ as a leaf,
or add the edge $\{b^1, a_{j_4}\}$ to $C$ and delete the edge $\{a_{j_2}, a_{j_1}\}$ to leave $a_{j_1}$ as a leaf.
Among these three trees, the maximum total weight of the internal vertices is at least
$a^3 + a^2 + a_{j_1} + a_{j_2} + b_{j_2} + a^1 + b^1 - \min\{b^1, a_{j_1}, b_{j_2}\} \ge 2 (b^3 + b^2 + \min\{b^1, a_{j_1}, b_{j_2}\})$;
thus, $C$ is settled.
This finishes the discussion on Case 7.

In summary, Cases 2-7 together prove that when $b^1 \ne b_{j_1}$, the component $C$ can be settled.
We next consider the situation where $b^1 = b_{j_1}$, that is, there is no type-V component attached to $b_{j_1}$.

\paragraph{Case 8.}
If $b^1 = b_{j_1}$ and $v = a_{j_2}$, then $d_G(b^1) = 2$ and thus $d_G(a^1) \ge 3$ by Operation~\ref{op1}.

{\bf Case 8.1.}
If there is a vertex $u$ outside $C$ that is adjacent to $a^1$, then we add the edge $\{a^1, u\}$ to $F$,
add the edge $\{a_{j_2}, b^1\}$ to $C$ while delete the edge $\{a_{j_2}, a^1\}$ from $C$;
this way, $C$ is settled.

{\bf Case 8.2.}
Note that if $a^1$ and $b_{j_3} \ (\ne b^3)$ are adjacent, then $a^1$ and $a_{j_3}$ are adjacent too.
By the claw-free property, we conclude that $a^1$ must be adjacent to a vertex $u \in \{a_{j_3}, b^3, b_{j_2}, b^2\}$.

{\bf Case 8.2.1.}
When $u = a_{j_3}$ (or $u = b_{j_2}$), we add the edges $\{a^1, u\}$ and $\{b^1, a_{j_2}\}$ to $C$
while delete the edges $\{u, a_{j_2}\}$ and $\{a_{j_2}, a^1\}$ from $C$;
this way, we obtain a path with two leaves $b^3$ and $b^2$, and thus settle $C$.

{\bf Case 8.2.2.}
Otherwise, $a_{j_3}$ and $b_{j_2}$ are adjacent.
When $u = b^2 \ (\ne b_{j_2})$, we add the edge $\{a^1, b^2\}$ to $C$
while delete the edge $\{a_{j_2}, a^1\}$ from $C$, to obtain a path with two leaves $b^3$ and $b^1$;
we may also add the edges $\{a^1, b^2\}$ and $\{a_{j_2}, b^1\}$ to $C$
while delete the edges $\{a_{j_2}, a^1\}$ and $\{a_{j_2}, b_{j_2}\}$ from $C$ to obtain a path with two leaves $b^3$ and $b_{j_2}$.
Between these two paths, the maximum total weight of the internal vertices is at least
$a^3 + a_{j_2} + b_{j_2} + a^1 + b^1 - \min\{b^1, b_{j_2}\} \ge 2 (b^3 + \min\{b^1, b_{j_2}\})$;
thus, $C$ is settled.
When $u = b^3 \ (\ne b_{j_3})$, the same argument applies to settle $C$.

\paragraph{Case 9.}
If $b^1 = b_{j_1}$ and $v = b_{j_2}$, we add the edge $\{b^1, b_{j_2}\}$ to $C$ while delete the edge $\{a_{j_2}, b_{j_2}\}$ from $C$;
this way, we obtain a path with two leaves $b^3$ and $b^2$, and thus $C$ is settled.

If $b^1 = b_{j_1}$ and $v = a_{j_3}$, we add the edge $\{b^1, a_{j_3}\}$ to $C$ while delete the edge $\{a_{j_3}, a_{j_2}\}$ from $C$;
this way, we obtain a path with two leaves $b^3$ and $b^2$, and thus $C$ is settled.

If $b^1 = b_{j_1}$ and $v = b_{j_3}$, then we distinguish whether $b_{j_3} = b^3$ or not.
When $b_{j_3} \ne b^3$, then $b^1$ and $a_{j_3}$ must be adjacent in $G$ and thus $C$ can be settled as in the last paragraph.
When $b_{j_3} = b^3$, if $a_{j_3}$ is adjacent to one of $a^1$ and $b_{j_2}$, then we can obtain a path with two leaves $b^3$ and $b^2$ to settle $C$;
if $a^1$ and $b_{j_2}$ are adjacent and $b_{j_2} \ne b^2$, then we can obtain a path with two leaves being $b^2$ and the lightest among $a^3, a_{j_2}, a^1$, 
and thus the total weight of the internal vertices of this path is at least
$a^2 + a^3 + a_{j_2} + a^1 + b^1 - \min\{a^3, a_{j_2}, a^1\} \ge 2 (b^2 + \min\{a^3, a_{j_2}, a^1\})$,
which settles $C$;
otherwise $b_{j_2} = b^2$, and then we add the edge $\{a^1, b^2\}$ to $C$ while delete the edge $\{a^2, a^1\}$ from $C$ to obtain a cycle,
followed by deleting the lightest edge of $C \cap M^* = \{\{a^1, b^1\}, \{a^2, b^2\}, \{a^3, b^3\}\}$ to settle $C$.

\paragraph{Case 10.}
If $b^1 = b_{j_1}$, $b^2 \ne b_{j_2}$, and $v$ is in the type-V component containing $b^2$, we distinguish whether $v$ is a head or a tail.
Note that the case where $b^3 \ne b_{j_3}$ and $v$ is in the type-V component containing $b^3$ can be argued in exactly the same way.

{\bf Case 10.1.}
If $v$ is a head, say $a_{j_4}$ of the edge $e_{j_4}$, assume that $a_{j_4}$ is adjacent to $b_{j_5}$ besides $b_{j_4}$.
Exactly the same argument as in Case 7.1 applies to settle $C$, since it does not matter whether $b^1 = b_{j_1}$ or not.

{\bf Case 10.2.}
If $v$ is a tail, say $b_{j_4}$ of the edge $e_{j_4}$.
Note that $b_{j_4} \ne b_{j_2}$, which has been dealt in Case 9.

{\bf Case 10.2.1.}
If $b^2 = b_{j_4}$, then besides the leaf $b^3$, we either do nothing to $C$ to leave $b^1, b^2$ as leaves,
or add the edge $\{b^1, b^2\}$ to $C$ and delete the edge $\{a_{j_2}, a^1\}$ to leave $a^1$ as a leaf,
or add the edge $\{b^1, b^2\}$ to $C$ and delete the edge $\{a_{j_2}, b_{j_2}\}$ to leave $b_{j_2}$ as a leaf.
Among these three trees, the maximum total weight of the internal vertices is at least
$a^3 + a^1 + a_{j_2} + b_{j_2} + b^1 + b^2 - \min\{b^1 + b^2, a^1, b_{j_2}\} \ge 2 (b^3 + \min\{b^1 + b^2, a^1, b_{j_2}\})$;
thus, $C$ is settled.

{\bf Case 10.2.2.}
If $b^2 \ne b_{j_4}$, then by the claw-free property and the definition of $v$ we conclude that $b^1$ is adjacent to $a_{j_4}$ too.
Assume that $a_{j_4}$ is adjacent to $b_{j_5}$ besides $b_{j_4}$.
When $b_{j_5} \ne b_{j_2}$, the argument in Case 7.1.1 can be applied to settle $C$;
when $b_{j_5} = b_{j_2}$,
if $a_{j_3}$ and $b_{j_2}$ are adjacent, then we add the edges $\{a_{j_3}, b_{j_2}\}$ and $\{b^1, a_{j_4}\}$ to $C$ and
delete the edges $\{a_{j_3}, a_{j_2}\}$ and $\{b_{j_2}, a_{j_4}\}$ to achieve a path with leaves $b^3$ and $b^2$;
if $a_{j_3}$ and $a^1$ are adjacent, then we add the edges $\{a_{j_3}, a^1\}$, $\{b^1, a_{j_4}\}$ and $\{b^1, b_{j_4}\}$ to $C$ and
delete the edges $\{a_{j_3}, a_{j_2}\}$, $\{a^1, b^1\}$ and $\{a_{j_4}, b_{j_4}\}$ to achieve a path with leaves $b^3$ and $b^2$;
if $b_{j_2}$ and $a^1$ are adjacent, then we add the edges $\{b_{j_2}, a^1\}$, $\{b^1, a_{j_4}\}$ and $\{b^1, b_{j_4}\}$ to $C$ and
delete the edges $\{a_{j_2}, b_{j_2}\}$, $\{a^1, b^1\}$ and $\{a_{j_4}, b_{j_4}\}$ to achieve a path with leaves $b^3$ and $b^2$.
Thus, $C$ is settled.

All possible cases have be discussed in the above.
The lemma is proven.
\end{proof}

\begin{lemma}
\label{lemma11}
A type-III component in the subgraph $H_3$ can be settled.
\end{lemma}
\begin{proof}
Recall that a type-III component $C$ in its original form contains only one edge $e_{j_1} = \{a_{j_1}, b_{j_1}\}$ of $M^*$,
with another edge $\{a_{j_1}, x\}$ where $x \in X$ (see Figure~\ref{fig03a});
there could be a type-V component attached to $b_{j_1}$, with its tail $b^1$ replacing the role of $b_{j_1}$.

\paragraph{Case 1.}
If $b^1$ is adjacent to a vertex $v$ outside $C$, then we add the edge $\{b^1, v\}$ to $F$, which settles $C$,
since the total weight of the internal vertices in $C$ is $w(C \cap M^*)$.

We next consider the case where $b^1$ is not adjacent to any vertex outside $C$, and thus it has to be adjacent to some vertex inside $C$.
Note that $C$ is a path with $b^1$ and $x$ being its two ending vertices.
Let $v$ denote the vertex adjacent to $b^1$ that is the \underline{\em farthest} to $b^1$ on $C$.
We distinguish this distance $d_C(b^1, v) \ge 2$ and where $v$ locates.

\paragraph{Case 2.}
If $b^1 \ne b_{j_1}$ and $v$ is in the type-V component containing $b^1$, then $v$ must be a tail of an edge of $M^*$ and thus $d_C(b^1, v)$ is even.
Denote this edge as $e_{j_2} = \{a_{j_2}, b_{j_2}\}$, that is $v = b_{j_2}$.

{\bf Case 2.1.}
If $d_C(b^1, v) \ge 4$, then denote the head vertex $b_{j_2}$ is adjacent to in the type-V component as $a_{j_3}$, besides $a_{j_2}$.
We conclude from the claw-free property that there must be at least an edge among $a_{j_2}, a_{j_3}, b^1$, which contradicts the identity of the type-V component.
Therefore, it is impossible to have $d_C(b^1, v) \ge 4$.

{\bf Case 2.2.}
If $d_C(b^1, v) = 2$, then we conclude that $d_G(b^1) = 2$ and thus $d_G(a^1) \ge 3$ by Operation~\ref{op1},
i.e. there is at least another edge incident at $a^1$ besides $\{a^1, b_{j_2}\}$ and $\{a^1, b^1\}$.
Denote this neighbor of $a^1$ as $u$, which is impossible to be inside $C$ by our construction algorithm.
Thus, we add the edge $\{b^1, b_{j_2}\}$ to $C$ while delete the edge $\{a^1, b_{j_2}\}$ from $C$,
and add the edge $\{a^1, u\}$ to $F$;
this way, the component becomes a path and thus $C$ is settled.

\paragraph{Case 3.}
If $b^1 \ne b_{j_1}$ and $v = b_{j_1}$, we consider the size of the type-V component containing $b^1$.

{\bf Case 3.1.}
If this type-V component contains more than one edge of $M^*$, then by the claw-free property $a_{j_1}$ must be adjacent to $b^1$,
which violates the definition of $v$ being the farthest and thus it is impossible.

{\bf Case 3.2.}
If the type-V component containing $b^1$ has only one edge of $M^*$, that is $\{a^1, b^1\}$,
then exactly the same argument in Case 2.2 settles $C$.

\paragraph{Case 4.}
If $b^1 \ne b_{j_1}$ and $v = a_{j_1}$, then we leave $C$ as it is if $w(b^1) \le w(b_{j_1})$,
or we add the edge $\{a_{j_1}, b^1\}$ to $C$ while delete the edge $\{a_{j_1}, b_{j_1}\}$ from $C$.
In either case, the total weight of the internal vertices in the resulting path is at least
$a_{j_1} + a^1 + \max\{b_{j_1}, b^1\} \ge 3 \min\{b_{j_1}, b^1\}$, and thus it settles $C$.

\paragraph{Case 5.}
If $b^1 \ne b_{j_1}$ and $v = x$, then we leave $C$ as it is if $w(b^1) \le w(a_{j_1})$,
or we add the edge $\{x, b^1\}$ to $C$ while delete the edge $\{a_{j_1}, x\}$ from $C$.
In either case, the total weight of the internal vertices in the resulting path is at least
$a^1 + \max\{a_{j_1}, b^1\} \ge 2 \min\{a_{j_1}, b^1\}$, and thus it settles $C$.

In summary, Cases 2-5 together prove that when $b^1 \ne b_{j_1}$, the component $C$ can be settled.
We next consider the situation where $b^1 = b_{j_1}$, that is, there is no type-V component adjacent to $b_{j_1}$.

\paragraph{Case 6.}
If $b^1 = b_{j_1}$, then $v = x$, and we conclude that $d_G(b^1) = 2$ and thus $d_G(a^1) \ge 3$ by Operation~\ref{op1},
that is, there is a vertex $u$ outside $C$ that is adjacent to $a^1$.
Thus we add the edge $\{a^1, u\}$ to $F$,
add the edge $\{x, b^1\}$ to $C$ while delete the edge $\{a^1, x\}$ from $C$;
this way, $C$ is settled.

All possible cases have be discussed in the above.
The lemma is proven.
\end{proof}

\begin{lemma}
\label{lemma12}
A type-IV component in the subgraph $H_3$ can be settled.
\end{lemma}
\begin{proof}
Denote the two edges of $M^*$ in the type-IV component $C$ in its original form as $e_{j_1}$ and $e_{j_2}$.
Note that both $a_{j_1}$ and $a_{j_2}$ are adjacent to a vertex $x \in X$ (see Figure~\ref{fig03b}),
and there could be a type-V component attached to $b_{j_1}$ and $b_{j_2}$, respectively, with the tails $b^1, b^2$ replacing the roles of $b_{j_1}, b_{j_2}$.
We assume w.l.o.g. that $w(b^1) \ge w(b^2)$.

\paragraph{Case 1.}
If $b^1$ is adjacent to a vertex $v$ outside $C$, then we add the edge $\{b^1, v\}$ to $F$, which settles $C$,
since the total weight of the internal vertices in $C$ is at least $a^1 + b^1 + a^2 \ge 3 b^2$.

We next consider the case where $b^1$ is not adjacent to any vertex outside $C$, and thus it has to be adjacent to some vertex inside $C$.
Note that $C$ is a path with $b^1$ and $b^2$ being its two ending vertices.
Let $v$ denote the vertex adjacent to $b^1$ that is the \underline{\em farthest} to $b^1$ on $C$.
We distinguish this distance $d_C(b^1, v) \ge 2$ and where $v$ locates.

\paragraph{Case 2.}
If $b^1 \ne b_{j_1}$ and $v$ is in the type-V component containing $b^1$, then $v$ must be a tail of an edge of $M^*$ and thus $d_C(b^1, v)$ is even.
Denote this edge as $e_{j_3} = \{a_{j_3}, b_{j_3}\}$, that is $v = b_{j_3}$.

{\bf Case 2.1.}
If $d_C(b^1, v) \ge 4$, then denote the head vertex $b_{j_3}$ is adjacent to in the type-V component as $a_{j_4}$, besides $a_{j_3}$.
We conclude from the claw-free property that there must be at least an edge among $a_{j_3}, a_{j_4}, b^1$, which contradicts the identity of the type-V component.
Therefore, it is impossible to have $d_C(b^1, v) \ge 4$.

{\bf Case 2.2.}
If $d_C(b^1, v) = 2$, then we conclude that $d_G(b^1) = 2$ and thus $d_G(a^1) \ge 3$ by Operation~\ref{op1},
i.e. there is at least another edge incident at $a^1$ besides $\{a^1, b_{j_3}\}$ and $\{a^1, b^1\}$.
Denote this neighbor of $a^1$ as $u$.
If $u$ is inside $C$, then $u = b^2$ by our construction algorithm and the claw-free property.
We add the edges $\{a^1, b^2\}$ and $\{b_{j_3}, b^1\}$ to $C$ while delete the edge $\{b_{j_3}, a^1\}$ from $C$ to obtain a cycle,
followed by deleting the lightest edge of $C \cap M^*$;
this settles $C$ since $C \cap M^*$ has at least three edges.
If $u$ is outside $C$, then we add the edge $\{b^1, b_{j_3}\}$ to $C$ while delete the edge $\{a^1, b_{j_3}\}$ from $C$,
and add the edge $\{a^1, u\}$ to $F$;
this way, the component becomes a path and thus $C$ is settled.

\paragraph{Case 3.}
If $b^1 \ne b_{j_1}$ and $v = b_{j_1}$, we consider the size of the type-V component containing $b^1$.

{\bf Case 3.1.}
If this type-V component contains more than one edge of $M^*$, then by the claw-free property $a_{j_1}$ must be adjacent to $b^1$,
which violates the definition of $v$ being the farthest and thus it is impossible.

{\bf Case 3.2.}
If the type-V component containing $b^1$ has only one edge of $M^*$, that is $\{a^1, b^1\}$,
then from the definition of $v$ we have $d_G(b^1) = 2$, and thus $d_G(a^1) \ge 3$ by Operation~\ref{op1},
i.e. there is at least another edge incident at $a^1$ besides $\{a^1, b_{j_1}\}$ and $\{a^1, b^1\}$.
The same argument as in Case 2.2, with $j_3$ replaced by $j_1$, applies to settle $C$.

\paragraph{Case 4.}
If $b^1 \ne b_{j_1}$ and $v = a_{j_1}$, then we leave $C$ as it is when $w(b^1) \le w(b_{j_1})$,
or we add the edge $\{a_{j_1}, b^1\}$ to $C$ while delete the edge $\{a_{j_1}, b_{j_1}\}$ from $C$.
In either case, the total weight of the internal vertices in the resulting path is at least
$a^2 + a_{j_1} + a^1 + \max\{b_{j_1}, b^1\} \ge 2 (b^2 + \min\{b_{j_1}, b^1\})$, and thus it settles $C$.

\paragraph{Case 5.}
If $b^1 \ne b_{j_1}$ and $v = x$, then by the definition of $v$ and the claw-free property $a_{j_1}$ and $b^1$ are adjacent.
We thus settle $C$ as in Case 4.
We note that this is simpler than Case 5 in the proof of Lemma~\ref{lemma10} because here $a_{j_1}$ and $a_{j_2}$ cannot be adjacent.

\paragraph{Case 6.}
If $b^1 \ne b_{j_1}$ and $v = a_{j_2}$, then we add the edge $\{a_{j_2}, b^1\}$ to $C$ while delete the edge $\{a_{j_2}, x\}$ from $C$.
This gives a path with leaves $b^2$ and $x$, and thus it settles $C$.

If $b^1 \ne b_{j_1}$ and $v = b_{j_2}$, then we conclude that $a_{j_2}$ and $b^1$ are adjacent when $b^2 \ne b_{j_2}$,
and thus we settle $C$ as in the last paragraph;
when $b^2 = b_{j_2}$, we add the edge $\{b^2, b^1\}$ to $C$ to obtain a cycle,
followed by deleting the lightest edge of $C \cap M^*$ from $C$.
Since $C \cap M^*$ has at least three edges, this settles $C$.

\paragraph{Case 7.}
If $b^1 \ne b_{j_1}$, $b^2 \ne b_{j_2}$, and $v$ is in the type-V component containing $b^2$, we distinguish whether $v$ is a head or a tail.

{\bf Case 7.1.}
If $v$ is a head, say $a_{j_3}$ of the edge $e_{j_3}$, and assume that $a_{j_3}$ is adjacent to $b_{j_4}$ besides $b_{j_3}$,
then we do nothing to $C$ to leave $b^1$ as a leaf when $w(b^1) \le w(b_{j_4})$,
or otherwise add the edge $\{a_{j_3}, b^1\}$ to $C$ while delete the edge $\{a_{j_3}, b_{j_4}\}$ from $C$ to leave $b_{j_4}$ as a leaf.
In either way, the total weight of the internal vertices of the resulting path is at least
$a^2 + a_{j_4} + a^1 + \max\{b_{j_4}, b^1\} \ge 2 (b^2 + \min\{b_{j_4}, b^1\})$.
Therefore, in either case $C$ can be settled.

{\bf Case 7.2.}
If $v$ is a tail, say $b_{j_3}$ of the edge $e_{j_3}$.
If $b^2 = b_{j_3}$, then we add the edge $\{b^1, b^2\}$ to $C$ to obtain a cycle,
followed by deleting the lightest edge of $C \cap M^*$ from $C$.
Since $C \cap M^*$ has at least three edges, this settles $C$.
If $b^2 \ne b_{j_3}$, then by the claw-free property and the definition of $v$ we conclude that $b^1$ is adjacent to $a_{j_3}$ too,
and thus the argument in Case 7.1 applies to settle $C$.
This finishes the discussion on Case 7.

In summary, Cases 2-7 together prove that when $b^1 \ne b_{j_1}$, the component $C$ can be settled.
We next consider the situation where $b^1 = b_{j_1}$, that is, there is no type-V component adjacent to $b_{j_1}$.

\paragraph{Case 8.}
If $b^1 = b_{j_1}$ and $v = x$, and we conclude that $d_G(b^1) = 2$ and thus $d_G(a^1) \ge 3$ by Operation~\ref{op1},
that is there is a vertex $u$ adjacent to $a^1$ other than $x$ and $b^1$.

{\bf Case 8.1.}
If $u$ is inside $C$, then $u = b^2$ by our construction algorithm and the claw-free property.
If $b^2 \ne b_{j_2}$, then we add the edges $\{a^1, b^2\}$ and $\{x, b^1\}$ to $C$ while delete the edge $\{x, a^1\}$ from $C$ to obtain a cycle,
followed by deleting the lightest edge of $C \cap M^*$;
this settles $C$ since $C \cap M^*$ has at least three edges.
If $b^2 = b_{j_2}$, then we add the edges $\{a^1, b^2\}$ and $\{x, b^1\}$ to $C$ while delete the edges $\{x, a^1\}$ and $\{x, a^2\}$ from $C$
to obtain a path with leaves $x$ and $a^2$ when $w(a^2) \le w(b^1)$,
or otherwise we add the edge $\{a^1, b^2\}$ to $C$ while delete the edge $\{x, a^1\}$ from $C$
to obtain a path with leaves $x$ and $b^1$.
In either way, the total weight of the internal vertices of the resulting path is at least
$a^1 + \max\{a^2, b^1\} \ge 2 \min\{a^2, b^1\}$.
Therefore, $C$ can be settled.

{\bf Case 8.2.}
If $u$ is outside $C$, then we add the edge $\{b^1, x\}$ to $C$ while delete the edge $\{a^1, x\}$ from $C$,
and add the edge $\{a^1, u\}$ to $F$;
this way, the component becomes a path and thus $C$ is settled.

\paragraph{Case 9.}
If $b^1 = b_{j_1}$ and $v = a_{j_2}$, then we add the edge $\{a_{j_2}, b^1\}$ to $C$ while delete the edge $\{a_{j_2}, x\}$ from $C$
to obtain a path with leaves $x$ and $b^2$.
Therefore, $C$ can be settled.

If $b^1 = b_{j_1}$ and $v = b_{j_2}$,
when $b^2 \ne b_{j_2}$ then we conclude from the claw-free property and the definition of $v$ that $a_{j_2}$ and $b^1$ are also adjacent,
and thus $C$ can be settled as in the last paragraph;
when $b^2 = b_{j_2}$, we can either do nothing to $C$ to leave $b^1$ and $b^2$ as leaves,
or add the edge $\{b^2, b^1\}$ to $C$ to obtain a cycle, followed by either deleting the edge $\{a^2, x\}$ to obtain a path with leaves $a^2$ and $x$,
or deleting the edge $\{a^1, x\}$ to obtain a path with leaves $a^1$ and $x$.
This way, the maximum total weight of the internal vertices among the three resulting paths is at least
$a^1 + a^2 + b^1 + b^2 - \min\{a^1, a^2, b^1 + b^2\} \ge 2 \min\{a^1, a^2, b^1 + b^2\}$.
Thus, $C$ can be settled.

\paragraph{Case 10.}
If $b^1 = b_{j_1}$, $b^2 \ne b_{j_2}$, and $v$ is in the type-V component containing $b^2$, we distinguish whether $v$ is a head or a tail.

{\bf Case 10.1.}
If $v$ is a head, say $a_{j_3}$ of the edge $e_{j_3}$, and assume that $a_{j_3}$ is adjacent to $b_{j_4}$ besides $b_{j_3}$,
then we do nothing to $C$ to leave $b^1$ as a leaf if $w(b^1) \le w(b_{j_4})$,
or otherwise add the edge $\{a_{j_3}, b^1\}$ to $C$ while delete the edge $\{a_{j_3}, b_{j_4}\}$ from $C$ to leave $b_{j_4}$ as a leaf.
This way, the total weight of the internal vertices of the resulting path is at least
$a^1 + a^2 + a_{j_4} + \max\{b^1, b_{j_4}\} \ge 2 (b^2 + \min\{b^1, b_{j_4}\})$.
Thus, $C$ can be settled.

{\bf Case 10.2.}
If $v$ is a tail, say $b_{j_3}$ of the edge $e_{j_3}$.
If $b^2 = b_{j_3}$, then we add the edge $\{b^1, b^2\}$ to $C$ to make it a cycle followed by deleting the lightest edge of $C \cap M^*$;
this settles $C$ since there are at least three edges in $C \cap M^*$.
If $b^2 \ne b_{j_3}$, then from the claw-free property and the definition of $v$, we conclude that $a_{j_3}$ must be adjacent to $b^1$ too.
Therefore, the same as in Case 10.1, $C$ can be settled.

All possible cases have be discussed in the above.
The lemma is proven.
\end{proof}

\begin{lemma}
\label{lemma13}
A type-VI component in the subgraph $H_3$ can be settled.
\end{lemma}
\begin{proof}
Recall that a type-VI component $C$ is a cycle containing two or more edges of $M^*$,
where the head of one edge of $M^*$ is adjacent to the tail of another edge of $M^*$ (see Figure~\ref{fig04b}).
Clearly, if there are three or more edges of $M^*$ in $C$, we simply delete the lightest one to settle $C$.
In the sequel we deal with the case where $C$ is a length-$4$ cycle.
Denote the two edges of $C \cap M^*$ as $\{a^1, b^1\}$ and $\{a^2, b^2\}$, and assume that $w(b^1) \ge w(b^2)$.

If $a^1$ ($a^2$, respectively) is adjacent to a vertex $v$ outside $C$,
then we add the edge $\{a^1, v\}$ ($\{a^2, v\}$, respectively) to $F$ and delete the edge $\{a^1, b^2\}$ ($\{a^2, b^2\}$, respectively) from $C$;
this way, the total weight of the internal vertices is at least
$a^1 + a^2 + b^1 \ge 3 b^2$,
and thus $C$ is settled.

If neither $a^1$ nor $a^2$ is adjacent to any vertex $v$ outside $C$, then we conclude from the construction algorithm that $d_G(a^1) = d_G(a^2) = 2$
since $a^1$ and $a^2$ are not adjacent to each other.
It follows from the claw-free property that neither $b^1$ nor $b^2$ can be adjacent to any vertex $u$ outside $C$.
This implies $|V| = 4$, a contradiction to our assumption that $|V| \ge 5$.
The lemma is proved.
\end{proof}

\begin{theorem}
\label{thm02}
The MwIST problem on claw-free graphs admits a $12/7$-approximation algorithm.
\end{theorem}
\begin{proof}
The above Lemmas~\ref{lemma09}--\ref{lemma13} state that every component of the subgraph $H_3 = G[V, M^* \cup M^{aa} \cup N^{aa} \cup N^{ax} \cup M^{ab} \cup N^{ab}]$
can be settled, without affecting any other components.
Also, such a settling process for a component takes only linear time, by scanning once the edges in the subgraph induced on the set of vertices of the component.
By settling, essentially the component is converted into a tree, possibly with one edge of $F$ specified for connecting a leaf of the tree outwards.

In the next step of the algorithm, it iteratively processes the heaviest component $C$, i.e. with the largest $w(C \cap M^*)$.
If the component $C$ has been associated with an edge $e$ of $F$,
and using the edge $e$ to connect a leaf of the resulting tree for $C$ outwards does not create a cycle,
then the algorithm does this and $C$ is processed.
This guarantees the total weight of the internal vertices in $V(C)$ is at least $2 w(C \cap M^*)/3$.
If using the edge $e$ to connect a leaf of the resulting tree for $C$ outwards would create a cycle,
the algorithm processes $C$ by replacing $C$ with another tree
that guarantees the total weight of the internal vertices in $V(C)$ at least $\frac 12 w(C \cap M^*)$.
Notice that the latter case happens only because of (at least) one edge of $F$ in an earlier iteration where a distinct component $C'$ was processed,
which connects a vertex of $C'$ into a vertex of $C$.
Therefore, every such $C$ is associated with a distinct component $C'$ processed by the algorithm in an earlier iteration, and thus $w(C') \ge w(C)$.
On the other hand, every such component $C'$ is associated to one $C$ only, due to its edge in $F$ connecting a leaf outwards into a vertex of $C$.
It follows that for this pair of components $C$ and $C'$, the total weight of the internal vertices in $V(C) \cup V(C')$ is at least
\[
w(C)/2 + 2 w(C')/3 \ge 7 (w(C) + w(C'))/12.
\]
After all components of $H_3$ are processed, we obtain a forest for which the total weight of the internal vertices therein is at least $7 w(M^*)/12$.
The algorithm lastly uses any other available edges of $E$ to interconnect the forest into a final tree, denoted as $T$;
clearly $w(T) \ge 7 w(M^*)/12$.

The time for the interconnecting purpose is at most $O(m \log n)$.
Therefore, by Corollary~\ref{coro01} we have a $7/12$-approximation algorithm for the MwIST problem on claw-free graphs.
\end{proof}

\section{Concluding remarks}
We have presented an improved approximation algorithm for the vertex weighted MIST problem, denoted MwIST,
which achieves the worst-case performance ratio $1/2$, beating the previous best ratio of $1/(3 + \epsilon)$,
designed by Knauer and Spoerhase in 2009~\cite{KS09}.
The key ingredient in the design and analysis of our algorithm is a novel relationship between MwIST and the maximum weight matching,
which we uncovered and it is inspired by the work \cite{LZ14, LZ14b, CHW16}.
A step further, for the problem restricted to claw-free graphs, we presented a $7/12$-approximation algorithm,
improving the previous best ratio of $1/2$ designed by Salamon in 2009 for claw-free graphs without leaves.
It would be interesting to see whether this newly uncovered relationship, possibly combined with other new ideas,
can be explored further to design better approximation algorithms for MwIST, or special cases of MwIST including claw-free graphs and cubic graphs.

\section*{Acknowledgement}
The authors are grateful to two reviewers for their insightful comments on the COCOON 2017 submission and
for their suggested changes that improve the presentation greatly.

ZZC was supported in part by the Grant-in-Aid for Scientific Research of the Ministry of Education, Culture, Sports, Science and Technology of Japan,
under Grant No. 24500023.
GL was supported by the NSERC Canada and the NSFC Grant No. 61672323;
most of his work was done while visiting ZZC at the Tokyo Denki University at Hatoyama.
LW was supported by the Hong Kong GRF Grants CityU 114012 and CityU 123013.
YC was supported in part by the NSERC Canada, the NSFC Grants No. 11401149, 11571252 and 11571087, and the China Scholarship Council Grant No. 201508330054.


\end{document}